\documentclass[journal]{IEEEtran}
\IEEEoverridecommandlockouts
\usepackage{cite}
\usepackage{amsmath,amssymb,amsfonts}
\allowdisplaybreaks
\usepackage{algpseudocode}
\usepackage{graphicx}
\usepackage{textcomp}
\usepackage[linesnumbered,ruled,vlined]{algorithm2e}
\usepackage{xcolor}
\usepackage{subcaption}
\usepackage{amsthm}
\usepackage{booktabs}
\usepackage{threeparttable}
\usepackage{framed}
\usepackage{float}
\usepackage{orcidlink}
\PassOptionsToPackage{hyphens}{url}\usepackage{hyperref}
\hypersetup{
  pdftitle={Incorporating Bounded Rationality into Electric Vehicle Highway Charging Decisions: A Bayesian Game Analysis},
  pdfauthor={Huanyu Yan; Xiaoying Tang},
  pdfkeywords={Electric vehicle, game theory, prospect theory, charging station, Nash equilibrium, bounded rationality, Bayesian game, range anxiety, highway charging, charging behavior prediction},
  colorlinks=true,
  linkcolor=blue,
  citecolor=blue,
  urlcolor=blue
}
\newtheorem{theorem}{Theorem}
\newtheorem{lemma}{Lemma}

\usepackage{parskip}
\usepackage{enumitem}

\usepackage{titlesec}
\titlespacing*{\section}
  {0pt}{5pt}{2pt}
\titlespacing*{\subsection}
  {0pt}{2pt}{1pt}

\def\BibTeX{{\rm B\kern-.05em{\sc i\kern-.025em b}\kern-.08em
    T\kern-.1667em\lower.7ex\hbox{E}\kern-.125emX}}

\newcommand{\coderepo}{\href{https://github.com/T-Lab-CUHKSZ/-IOTJ25-Incorporating-Bounded-Rationality-into-Electric-Vehicle-Highway-Charging-Decisions}{T-Lab-CUHKSZ/IOTJ25 on GitHub}}

\begin{document}
\title{Incorporating Bounded Rationality into Electric Vehicle Highway Charging Decisions: A Bayesian Game Analysis\\
\thanks{
\par This is the author's version accepted for publication in \textit{IEEE Internet of Things Journal}, vol. 12, no. 11, pp. 15249--15260, 2025. The final published version is available at \url{https://doi.org/10.1109/JIOT.2025.3530449}. The source code is publicly available at \coderepo.
\par This work was supported in part by the National Natural Science Foundation of China under
Grant 62001412; in part by the Shenzhen Institute of Artificial Intelligence
and Robotics for Society; in part by the Shenzhen Key Laboratory of
Crowd Intelligence Empowered Low-Carbon Energy Network under Grant
ZDSYS20220606100601002; in part by the Shenzhen Stability Science
Program 2023; and in part by the Guangdong Provincial Key Laboratory
of Future Networks of Intelligence under Grant 2022B1212010001. This
article was presented in part at the 14th ACM International Conference
on Future Energy Systems and was published in its proceedings \cite{yan2023incorporating}. (Corresponding author: Xiaoyu Tang.).
\par The authors are with the School of Science and Engineering, The Chinese University of Hong Kong, Shenzhen, Guangdong, 518172, P.R. China, and the Shenzhen Institute of Artificial Intelligence and Robotics for Society, and the Shenzhen Key Laboratory of Crowd Intelligence Empowered Low-Carbon Energy Network. (e-mail: huanyuyan@link.cuhk.edu.cn;tangxiaoying@cuhk.edu.cn).
}
}

\author{Huanyu~Yan~\orcidlink{0000-0001-9619-066X},~\IEEEmembership{Student Member,~IEEE},
        and~Xiaoying~Tang~\orcidlink{0000-0003-3955-1195},~\IEEEmembership{Member,~IEEE}}

\maketitle

\begin{abstract}
Electric vehicles (EVs) represent a critical intelligent terminal within the Internet of Things (IoT). Despite the year-on-year growth in EV penetration, the highway driving experience still requires improvement. Accurate prediction of EV highway charging behavior is crucial to addressing this issue. This paper introduces a novel bounded rationality framework to analyze highway charging decisions. Specifically, we utilize prospect theory to capture the tendency of drivers to reserve more electricity than theoretically necessary. We then propose a Bayesian game in which EV drivers, unaware of others’ decisions, aim to minimize costs, including range anxiety, charging fees, and queuing time. To gain insights into the game, we prove the existence and uniqueness of the Bayesian Nash Equilibrium in two practical scenarios. Our numerical experiments, based on real-life data, demonstrate that drivers’ risk aversion tendency significantly influence EV charging decisions, charging demand, queuing lengths at charging stations, and the departure rate on the highway network. Furthermore, our strategy reduces \textcolor{black}{cumulative EV cost and CSs' charging costs} compared to other benchmarks.
\end{abstract}


\setlength{\parskip}{0pt}
\setlength{\parindent}{15pt}
\setlength{\abovedisplayskip}{5pt}
\setlength{\belowdisplayskip}{0pt}

\section{Introduction}

Electric vehicles (EVs) represent a critical intelligent terminal within the Internet of Things (IoT), bridging the transportation and power networks \cite{yu2014phev}. Despite the increasing penetration of EVs over recent years, long-distance driving remains a significant challenge, particularly due to severe congestion at charging stations (CSs) \cite{Longqueuingtimenews_China, Longqueuingtimenews_Australia}. Addressing this congestion and improving the highway driving experience for EV users necessitates a thorough analysis of EV owners’ charging decisions. This analysis is a foundational step towards resolving broader issues related to EV highway charging, including optimal charging station placement, the development of effective pricing strategies, and the efficient scheduling of EV charging.

\begin{figure*}[t]
  \centering
  \includegraphics[width=\linewidth]{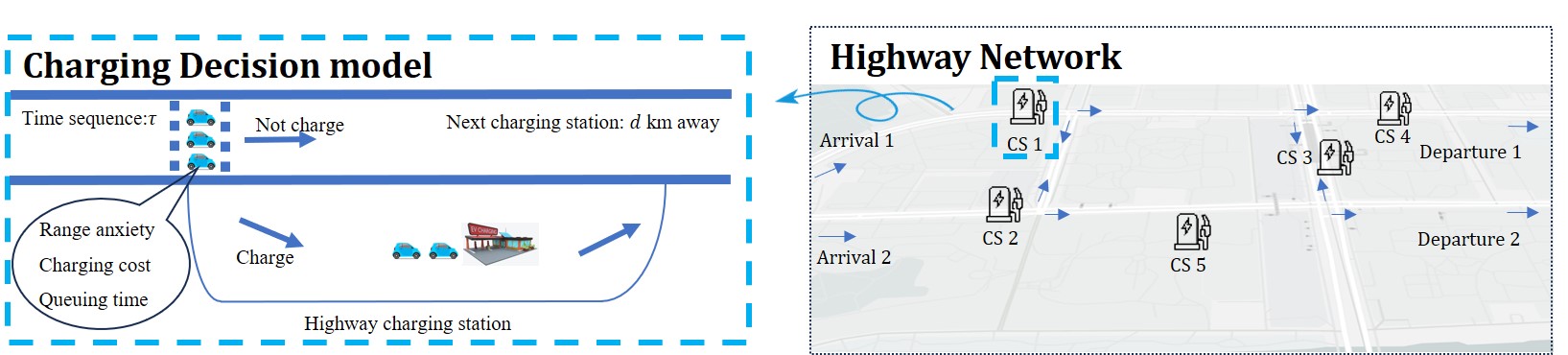}
  \caption{System overview of bounded rationality based highway charging Bayesian game: EV owners make charging decisions at each CS based on known road network information and the behavior analysis of other vehicles. These decisions impact charging demand, queue lengths at charging stations, departure rates, and average travel times on the road network.}
  \label{fig:system}
\end{figure*}

Despite behavior of EV charging has been extensively studied in the literature \cite{silva2024stochastic,hu2019modeling,yin2024research,williams2024driving,yang2016predicting,hu2018analyzing,xu2017joint}, numerous studies \cite{williams2024driving,silva2024stochastic,yin2024research} do not account for bounded rationality, which has been shown to significantly influence highway charging behavior \cite{yang2016predicting,xu2017joint,hu2019modeling}. Among the works that consider bounded rationality, \color{black} Xu \textit{et al.} analyzed data on EV charging preferences in Japan \cite{xu2017joint}, revealing that a user's bounded rationality behavior, particularly range anxiety, leads owners to charge more frequently than necessary. However, their study did not model bounded rationality in charging decisions. Additionally, \cite{yang2016predicting} and \cite{hu2018analyzing} proposed residual models on bounded rationality. In these models, drivers charge their vehicles if the state-of-charge is below the sum of theoretical need and constant redundancy. The uniform value of redundancy across all vehicles lack refined modeling for different users. Therefore, there remains a need for accurate modeling of bounded rationality in EV charging. \color{black}

Prospect theory is a prominent framework for portraying finite rationality and has been applied in various engineering models \cite{wang2015load,li2012prospects,liao2019prospect,sanjab2017prospect,saad2017hardware,mediwaththe2018game}. For instance, Wang \textit{et al.} leveraged prospect theory to explore grid customers’ willingness to participate in demand-side management \cite{wang2015load}. Similarly, Liao \textit{et al.} applied prospect theory in the design of a privacy-preserving mechanism to analyze individuals’ decisions to contribute data when confronted with privacy concerns \cite{liao2019prospect}. Nevertheless, the utility designs in these studies differ markedly from the highway charging scenario and thus cannot be directly applied. \color{black} Hu \textit{et al.} applied prospect theory to capture the bounded rationality of one single EV in \cite{hu2019modeling}, where the EV calculates the decision variable at the next station with prospect theory and compares it to a threshold based on remaining charging opportunities. Charging occurs when this variable exceeds the threshold. However, this model only considers individual paths, ignoring the impact of other EVs' decisions on queuing congestion. This oversight means the model may be inaccurate in systems involving multiple vehicles, as the charging decisions of other EVs affect queuing lengths and, in turn, individual charging decisions. To address this gap, we propose a novel prospect theory-based model that analyzes charging decisions involving multiple heterogeneous EVs in a highway road network with continuous traffic flow, aiming to minimize charging costs, queuing time, and anxiety from bounded rationality.
\color{black}

Applying prospect theory to heterogeneous EV systems presents several challenges. The utility function in prospect theory is derived from real-life experiments, and its non-linearity and exponential parameters complicate the analysis of mathematical properties, such as the existence and uniqueness of equilibrium. Additionally, in the highway charging scenario, a critical parameter—the state-of-charge of each EV—is known only to the driver. This lack of information further complicates model development. Our main contributions are summarized as follows:

\begin{enumerate}[topsep=0pt, partopsep=0pt]
    \item \textbf{A Novel Bounded Rational Framework of EV Highway Charging Behaviors:} In this paper, we propose a novel prospect theory-based framework to analyze the bounded rationality of EVs’ highway charging behaviors. We capture the bounded rational phenomenon that drivers often leave more power than they theoretically need to ensure safety on the highway. To model the charging interactions among EVs, we introduce a prospect theory-based game, considering drivers’ range anxiety, charging prices, and queuing times.
    \item \textbf{Theoretical Analysis on Equilibrium Properties:} To analyze the insights, we rigorously prove the mathematical form of Bayesian Nash Equilibrium. Then we prove the existence and the uniqueness of the Bayesian Nash equilibrium under two practical scenarios, where with and without the EV destination distribution knowledge. Algorithms are designed to solve the charging decisions based on the equilibrium properties.
    \item \textbf{Risk Aversion Effects on Highway Network:} Using the real highway arrival data, we experimentally verified that the EV charging decisions, the charging demand and queuing length of charging stations, and the departure rate of highway network differ significantly under the different drivers' risk aversion tendencies. Further, we used our method to solve the optimal power purchase strategy for charging stations in the day-ahead market, which successfully reduces the power purchase cost. This indicates that our model serve as a foundational step towards addressing further EV highway charging issues, including the selection of charging station locations, scheduling of EV charging, pricing strategies, etc.
\end{enumerate}

\section{Bounded Rational Highway Charging Model}
\label{sec: model}

We consider a highway charging decision problem as depicted in Figure~\ref{fig:system}. A day’s arrivals are divided into a continuous sequence of small time periods. At any time $\tau$, a small segment of EVs arrives at a service area with a charging station. The drivers are aware that the distance to the next service area with a charging station is $d$. Meanwhile, some EVs have already arrived and are queuing, with the estimated queuing time for these vehicles denoted by $t_0$. For the EVs in this segment, drivers decide whether to charge based on the state-of-charge (SoC) of their vehicles. If the drivers choose to charge, they may face potentially long queuing times and high charging fees, as some stations may increase prices when demand is high. Conversely, if the drivers choose not to charge, they may experience range anxiety en route to the next charging station.

In this paper, we analyze the highway charging decision problem considering the queuing time, the charging fee, and the range anxiety. The details are as follows.

\subsection{Range Anxiety}

Range anxiety is an important reason for bounded rationality. In this section, we introduce a range anxiety model, grounded in the framing effect of Prospect Theory \cite{kahneman1979prospect}. Prospect theory, a behavioral model, elucidates decision-making processes under risk, highlighting general risk aversion and the influence of expectations on perception. The framing effect, a cognitive bias, posits that decision-making is swayed by the framing of outcomes relative to expectations, which serve as reference points. Deviations from these reference points significantly impact cognition. As depicted in Figure~\ref{fig:framing_effect}, cognition escalates more rapidly as values diverge from the reference point, with a steeper curve for negative values, reflecting risk aversion.

In the highway charging scenario, we define two thresholds to capture different drivers’ frames: $SoC_{s,i}$ and $SoC_{d,i}$. $SoC_{s,i}$ is the safe threshold for driver $i$, who perceives that an SoC exceeding $SoC_{s,i}$ ensures safe arrival at the next charging station, even under worst-case conditions. Conversely, $SoC_{d,i}$ is the danger threshold, below which driver $i$ deems arrival at the next station impossible. Given prior research indicating a near-linear decrease in SoC with travel distance at constant speed \cite{xu2020mitigate,pelletier2017battery,xu2019fleet}, we model the relationship between these thresholds and distance $d$ as follows:
\begin{equation} \label{eq:dsi}
\small
  {SoC}_{s,i}=\left\{
  \begin{aligned}
    & \frac{1}{d_{s,i}}d_i, & d_i<d_{s,i}, \\
    & 1, & d_i \geq d_{s,i}; \\
  \end{aligned}
  \right.
\end{equation}
\begin{equation} \label{eq:d0i}
\small
  {SoC}_{d,i}=\left\{
  \begin{aligned}
    & \frac{1}{d_{0,i}}d_i, & d_i<d_{0,i}, \\
    & 1, & d_i \geq d_{0,i}, \\
  \end{aligned}
  \right.
\end{equation}
where $d_{s,i}$, $d_{0,i}$ is the absolutely safe and theoretical maximum distance the vehicle can travel when fully charged. Figure~\ref{fig:thresholds} provides an illustrative example. Consider an EV driver who estimates the maximum and minimum distances that a fully charged EV can travel as 600km and 400km, respectively. At a distance of 100km, his $SoC_s$ is 25\% and $SoC_d$ is around 16.7\%. Utilizing these thresholds, we formulate the range anxiety model $A_i$ as follows:

\begin{figure}[t]
    \centering
    \begin{subfigure}{0.49\columnwidth}
        \centering
        \includegraphics[width=\linewidth]{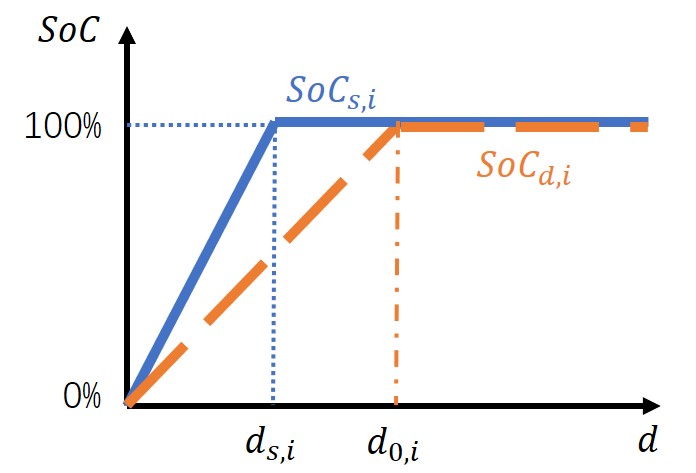}
        \caption{Two thresholds.}
        \label{fig:thresholds}
    \end{subfigure}
    \hfill
    \begin{subfigure}{0.49\columnwidth}
        \centering
        \includegraphics[width=\linewidth]{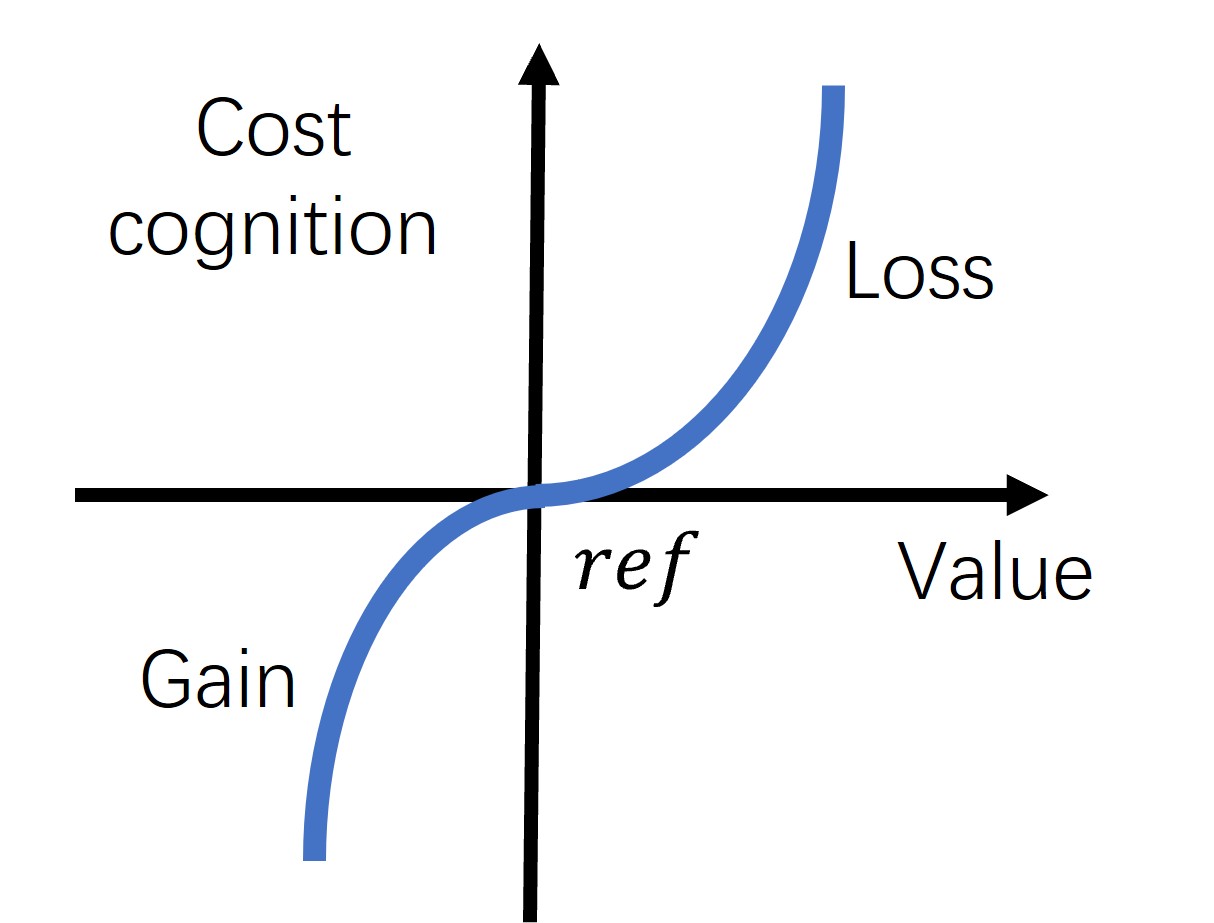}
        \caption{\textcolor{black}{Framing effect utility.}}
        \label{fig:framing_effect}
    \end{subfigure}
    \caption{An intuitive description of two thresholds \textcolor{black}{$SoC_{s,i}$} and \textcolor{black}{$SoC_{d,i}$}, and utility function of framing effect in prospect theory.}
    \label{fig:lambda_varies}
\end{figure}

\begin{equation} \label{eq:a_i}
\small
   A_{i}=\left\{
  \begin{aligned}
    & 0, & \text{if } {SoC}_i \geq {SoC}_{s,i}, \\
    & \lambda_i({SoC}_{s,i}-{SoC}_i)^{\alpha_i}, & \text{if } {SoC}_{d,i} \leq {SoC}_i \leq {SoC}_{s,i}, \\
    & \infty, & \text{otherwise},
  \end{aligned}
  \right.
\end{equation}
where $\lambda_i$ and $\alpha_i$ are the rationality parameters representing the loss aversion degree and the convexity of the value function respectively. For a perfectly rational person, $\lambda = 0$ and $\alpha = 1$, but it's hard for people to be perfectly rational \cite{kahneman1979prospect}. According to survey \cite{tanaka2010risk,xu2020mitigate}, $\lambda>1$ for the individual is more loses (risk) aversion than gains, and $\alpha > 1$ for anxiety escalates faster as the driver’s SoC deviates further from the safe value. \color{black} It also coincides with the conclusion in the literature \cite{saad2016toward,liao2019prospect,tanaka2010risk} that cognitive to loss is convex.  \color{black} Intuitively, $\lambda_i > \lambda_j$ if EV $i$’s owner is more risk-averse than EV $j$’s. The convexity, $\alpha$, mirrors how individuals value outcomes relative to their expectations. For example, if both drivers’ state-of-charge is 5\% below their safe threshold, and driver $i$ exhibits greater nervousness than driver $j$, then $\alpha_i > \alpha_j$. People's rationality parameters can be obtained via survey methods \cite{tanaka2010risk}. This range anxiety model aligns with prior survey-based studies, indicating zero anxiety under sufficient SoC \cite{guo2018battery}, and more pronounced variability as SoC decreases \cite{xu2020mitigate}.

In addition, we denote \textcolor{black}{$\mu_{a,i}$} as the anxiety cost parameter. It presents the cost of a vehicle breaking down on the highway and therefore should be a large number. With above notations, the anxiety utility can be modeled as follows,

\begin{equation}
  \Pi_{a,i}=\left\{
  \begin{aligned}
    & \textcolor{black}{-\mu_{a,i}} A_i & \text{if not charge}, \\
    & 0. & \text{otherwise}. \\
  \end{aligned}
  \right.
\end{equation}

\subsection{Queuing Time}

In this subsection, we addresses the estimation of queuing time for an EV driver. The driver’s estimation takes into account the queuing time of previously arrived EVs and the decisions of other EVs in the same flow segment. We denote the charging power as $\eta$ and number of charging piles as $k$. The charging time for EV $i$ , with a state-of-charge $SoC_i$,  is computed as $\mathord{V_i(\textcolor{black}{SoC_{t,i}}-{SoC}_i)} / \mathord{\eta}$, where the $\textcolor{black}{SoC_{t,i}}$ denotes the target SoC to which the EV will charge \color{black}, $V_i$ denotes the battery volume of EV $i$. \color{black}

We define $c(N)$ as a function mapping the charging decisions of all $N$ EV drivers to the count of EVs opting to charge in this EV flow segment ($0\leq c(N)\leq N$). Since the time period is small, each charging EV is assumed to have an equal probability of occupying any specific position in the queue following the previously queuing EVs. However, as the driver is unaware of the state-of-charge of the other EVs, they \textcolor{black}{are assumed to} estimate each EV’s queuing time as $\mathbb{E}[\mathord{V(\textcolor{black}{SoC_{t,i}}-\textcolor{black}{SoC_i})}/ \mathord{\eta}]$. If EV $i$ is at position $j$ in the queue, its expected waiting time is $(j-1)\mathbb{E}[\mathord{V(\textcolor{black}{SoC_{t,i}}-\textcolor{black}{SoC_i})} / \mathord{\eta}]$. Consequently, for any driver opting to charge, the queuing time can be estimated as follows:

\begin{align} \label{eq:queuing_time}
\small
  t &= \frac{1}{k} (\sum_{j=1}^{c(N)}\frac{1}{c(N)}(j-1)\mathbb{E}[\frac{V(\textcolor{black}{SoC_{t,i}}-\textcolor{black}{SoC_i})}{\eta}]) + t_0, \nonumber \\
  &= \frac{1}{2k} (c(N)-1) \mathbb{E}[\frac{V(\textcolor{black}{SoC_{t,i}}-\textcolor{black}{SoC_i})}{\eta}] + t_0,
\end{align}

where $t_0$ denotes the queuing time caused by queuing EVs before this time sequence. In this paper, we assume $t_0$ is a constant known to all players since some charging station operators will announce the estimated waiting time on APP or website. Besides, EV drivers can observe the waiting queue to estimate $t_0$.

Then, the congestion penalty can be modeled as follows,

\begin{equation}
\small
  \Pi_{q,i}=\left\{
  \begin{aligned}
    & -\textcolor{black}{\mu_{q,i}} \cdot t & \text{if charge}, \\
    & 0, & \text{otherwise}, \\
  \end{aligned}
  \right.
\end{equation}

where $\mu_{q,i}$ is a time-value parameter of EV $i$. Additionally, when the state-of-charge is approaching to the dangerous threshold, we assume the driver will definitely prefer queuing to charge, rather than taking the risk of breaking down on the highway. Mathematically,
\begin{align} \label{eq:mu_a_much_greater}
    lim_{SoC_{i} \rightarrow SoC_{d,i}^+} \textcolor{black}{\mu_{a,i}} A_i \gg \textcolor{black}{\mu_{q,i}} t.
\end{align}

\subsection{Charging Cost}

Since some highway charging stations set a high service price when the demand is large, we consider the charging cost in the drivers' utility. We denote the charging unit price as $P$ CNY/kWh. In this paper, we assume the charging price won't change during EV drivers' decision period. Then, the charging cost for EV $i$ could be written as
\begin{equation}
\small
  \Pi_{c,i}=\left\{
  \begin{aligned}
    & -V_i (\textcolor{black}{SoC_{t,i}}-{SoC}_i) \cdot P, & \text{if charge}, \\
    & 0, & \text{otherwise}. \\
  \end{aligned}
  \right.
\end{equation}

\subsection{Highway Charging Bayesian Game Model}

To sum up, we consider a segment of EV flow including $N$ selfish EV arriving at the service area with a charging station. Each EV driver decides whether to charge at the current charging station. All players have the same action space, denoted by $X_i = \{0,1\}$, $i\in \mathcal{N}$, where action $x_i=1$ denotes EV decides to charge at the current charging station, and action $x_i=0$ denotes not to charge. The strategy of all $N$ players can be represented as a vector $\boldsymbol{x}=[x_i]_{\forall i \in \mathcal{N}} \in \boldsymbol{X}$, where $\boldsymbol{X}=X_1 \times X_2 \times \cdots \times X_N$ is the set of all possible strategies. The utility function of EV drivers is as follows.
\begin{equation}
\label{eq:utility}
\begin{aligned}
  \Pi_i(\boldsymbol{x}) &= \Pi_{a,i}+\Pi_{c,i}+\Pi_{q,i} \\
  &=\left\{
  \begin{aligned}
    & -V_i (\textcolor{black}{SoC_{t,i}}-{SoC}_i) \cdot P -\textcolor{black}{\mu_{q,i}} \cdot t, & \text{if $x_i=1$}, \\
    & -\textcolor{black}{\mu_{a,i}}A_i, & \text{if $x_i=0$}.
  \end{aligned}
  \right.
\end{aligned}
\end{equation}

Given the fact that the charging time for an EV is quite long, the charging decisions of other vehicles significantly influence the queuing time and subsequent charging decisions of other EVs. Consequently, we formulate the highway charging decision problem as a game. To capture the feature that an EV driver is unaware of the SoC of other EVs, we employ the Bayesian Game, treating SoC as \textit{private information}. Each player lacks knowledge of others’ private information but is presumed to understand the distribution of such information, termed as \textit{common knowledge}. We denote that distribution's density function as $f(SoC)$. With these notations, we model the Highway Charging Game as follows.
\begin{framed}
\textit{\textbf{Highway Charging Bayesian Game}}
\begin{itemize}
\item \textit{Players}: $N$ EVs.
\item \textit{Private Information}: The state-of-charge of each EV.
\item \textit{Common Knowledge}: In this small segment of EV flow, the distribution density function of private information is $f(SoC)$.
\item \textit{Strategy}: Charge or not at the current charging station.
\item \textit{Payoffs}: The payoff of each EV $i$ with private information $SoC$ as defined in \eqref{eq:utility}.
\end{itemize}
\end{framed}

Beyond private information, the accessibility of information among other electric vehicles is also worth discussing. Some EV owners can estimate the destination distribution of this segment of EV flow based on experience, while others do not. \color{black} Given that the driver's knowledge of that destination distribution will significantly change the predicted outcome, we consider the two scenarios separately. \color{black} For those EVs whose owners without destination distribution knowledge, we assume they adopt self-referencing \cite{burnkrant1995effects}, a psychological effect where individuals use their own situation as a reference when estimating others’ circumstances \cite{cacioppo1979effects}. Thus, these drivers are presumed to perceive other EVs’ parameters as similar to their own, and we will discuss this scenario in Section~\ref{ssec:no_distribution}. Subsequently, we will examine the case where EV owners are aware of the destination distribution in Section~\ref{ssec:with_distribution}. The differences in information access are summarized in Table~\ref{tab:para_two_settings}.

\begin{table}[t]
\small
    \caption{Information access under two settings. The setting 1 is without the destination distribution knowledge discussed in Section~\ref{ssec:no_distribution}, and setting 2 is with the destination distribution discussed in Section~\ref{ssec:with_distribution}.}
    \label{tab:para_two_settings}
    \centering
    \begin{threeparttable}
    \begin{tabular}{cccc}
        \toprule
        Types & Description/Notation & Setting 1 & Setting 2 \\
        \midrule
        EV Information & $V_i$ & $\square$ & $\square$ \\
        & $d_{s,i}$ & $\square$ & $\square$ \\
        & $d_{0,i}$ & $\square$ & $\square$ \\
        & $\lambda_i$ & $\square$ & $\square$ \\
        & $\alpha_i$ & $\square$ & $\square$ \\
        & Destination & $\square$ & $\blacksquare$ \\
        CS Information & Positions & $\bullet$ & $\bullet$ \\
                       & $P$ & $\bullet$ & $\bullet$ \\
                       & $t_0$ & $\bullet$ & $\bullet$ \\
                       & Arrival $SoC$ & $\blacksquare$ & $\blacksquare$ \\
        \bottomrule
    \end{tabular}
    \begin{tablenotes}
        \item[1] $\bullet$ denotes the drivers know the exact value.
        \item[2] $\square$ denotes the drivers do not know the exact value and evaluate it by self-referencing.
        \item[3] $\blacksquare$ denotes the drivers do not know the exact value but know the distribution of it.
    \end{tablenotes}
    \end{threeparttable}
\end{table}

\section{Equilibrium of Highway Charging Bayesian Game}
\label{sec:solution}

In the Bayesian game content, the most important solution concept is the \textit{Bayesian Nash Equilibrium} (BNE). A BNE is defined as a strategy profile that maximizes the expected payoff for each player given their beliefs about other players' private information and the strategies played by the others. It provides coalitional stability in that once the game fills into the equilibrium, no player can increase his expected profit if unilaterally deviates from the equilibrium. Mathematically, the BNE is the strategy vector with the following property:
\begin{equation} \label{eq:BNE_definition}
    x_i^*(\mathbf{x_{-i}^*}) = \mathop{\arg\max}\limits_{x_i}\mathbb{E}[\Pi_i(x_i,\mathbf{x_{-i}^*})], \forall i \in \mathcal{N}.
\end{equation}

\subsection{Existence and Uniqueness of BNE}
\label{subsec:existence_uniqueness}

In the non-cooperative game, the existence of the Nash Equilibrium is not always guaranteed \cite{bacsar1998dynamic}. Therefore, it is necessary to investigate the existence of the  BNE. In this section, we will discuss the existence and uniqueness of BNE in both scenarios. \textcolor{black}{To simplify the discussion, we assume that the weighting parameters and target SoC for all vehicles are the same, denoted as $\mu_a$,$\mu_q$ and $SoC_t$.\footnote{\textcolor{black}{We will discuss the diverse parameters scenario in Section~\ref{sec:diverse_para}.}}}

According to the definition of BNE, the driver with state-of-charge $SoC_i$ choose the higher expected profits between charge $\mathbb{E}[\Pi(x_i=1)|SoC_i]$ and not to charge $\mathbb{E}[\Pi(x_i=0)|SoC_i]$. Therefore, we define a function $h(SoC_i)$ as
\begin{align} \label{eq:h}
\small
    h(SoC_i) :=& \mathbb{E}[\Pi(x_i=0)|SoC_i] - \mathbb{E}[\Pi(x_i=1)|SoC_i], \nonumber \\
    =& \mathbb{E}[-\mu_aA_i] - \mathbb{E}[-V_iP(SoC_t-SoC_i) -\mu_q \cdot t].
\end{align}

Given the state-of-charge of driver known, the range anxiety and charging price are constants. Therefore,
\begin{align}
    h(SoC_i) = -\mu_aA_i + V_iP(SoC_t-SoC_i) + \mu_q \mathbb{E}[t].
\end{align}

According to Equation~\eqref{eq:queuing_time}, the expected queuing time $\mathbb{E}[t]$ is influenced by the count of charging EVs $c(N)$, which is hard to obtain the close-form expression. However, the expected queuing time is independent with the remaining energy level \textcolor{black}{$SoC_i$. For example, whether} an EV arrives at the charging station with 20\% or 50\% energy left, the expected queuing time doesn't change. This indicates $\frac{\rm d \mathbb{E}[t]}{\rm d SoC_i}=0$. Leveraging this finding, we examine the deviation of function $h(SoC_i)$, leading to the subsequent lemma:

\begin{lemma} \label{lmm:NE_form}
In both scenarios of the highway charging Bayesian game, the charging decision of EV $i$ under BNE $x_i$ is as follows,
\begin{align} \label{eq:BNE_form}
x_i=1 \text{ if } SoC_i < \widehat{SoC_i}, \nonumber \\
x_i=0 \text{ if } SoC_i \ge \widehat{SoC_i},
\end{align}
where $\widehat{SoC_i}$ is a constant between 0\% to 100\%.
\end{lemma}

The proof of Lemma~\ref{lmm:NE_form} is placed in Appendix~\ref{apdx-prof_NE_form}. Intuitively, in the BNE, each vehicle has a charging decision point $\widehat{SoC_i}$ and charges when and only when the energy is less than that decision point. Lemma~\ref{lmm:NE_form} shows that the decision point of each EV exists and is unique, which also proves the solution uniqueness in $N$-EV system as follows.

\begin{theorem}\label{thm:existence_uniqueness}
In both scenarios of the highway charging Bayesian game, there exists a unique BNE under any parameters.
\end{theorem}

\subsection{BNE Solving Method without Destination Distribution Knowledge}
\label{ssec:no_distribution}

In this subsection, we design solution for BNE without destination distribution knowledge. According to Equation~\eqref{eq:queuing_time},
\begin{align} \label{eq:Et}
\small
\mathbb{E}[t] &= \mathbb{E}[\frac{1}{2k} (c(N)-1) \mathbb{E}[\frac{V(SoC_t-\textcolor{black}{SoC_i})}{\eta}] + t_0], \nonumber \\
     &= \frac{1}{2k} (\mathbb{E}[c(N)]-1) \frac{\mathbb{E}[V(SoC_t-\textcolor{black}{SoC_i})]}{\eta} + t_0.
\end{align}

We first consider $\mathbb{E}[V(SoC_t-\textcolor{black}{SoC_i})]$ in \eqref{eq:Et}. For any vehicle, the battery capacity $V$ is independent of remaining energy level $SoC$, therefore
\begin{align}
    \mathbb{E}[V({SoC}_t-{SoC})] = \mathbb{E}[V] (SoC_t - \mathbb{E}[{SoC}]).
\end{align}
Since the distribution of state-of-charge of arrival EVs is the common knowledge, $\mathbb{E}[{SoC}] = \int f(SoC)\mathrm{d}SoC$. Besides, under the self-referencing, drivers consider $E[V]=V_i$. Thus
\begin{align}
    \mathbb{E}[V(SoC_t-{SoC})] = V_i (SoC_t - \int f(SoC)\mathrm{d}SoC).
\end{align}

We next examine the expected count of charging EVs, $\mathbb{E}[c(N)]$ in \eqref{eq:Et}. For an EV characterized by parameters $(V_i,d_i,d_{s,i},d_{0,i})$ and a random state-of-charge following distribution $f(SoC)$ in a system of $N$ EVs, the optimal strategy could be to charge or not. Given the stochastic state of charge, we denote the probability that the optimal decision is to charge as $p_i$. Under the self-referencing, a driver with parameters $(V_i,d_i, d_{s,i},d_{0,i})$ considers the charging probability of the other N-1 EVs as $p_i$ when evaluating the cost of charging.\footnote{We will release this assumption that EVs infer other EVs' destination by self-referencing in Section~\ref{ssec:with_distribution}.} Thus, $\mathbb{E}[c(N)]$ is the expected value of one (himself) plus the Binomial distribution with $N-1$ trials and success probability $p_i$, i.e.,
\begin{align}
\small
   \mathbb{E}[c(N)] = 1+(N-1)p_i.
\end{align}
With the BNE form result in Lemma~\ref{lmm:NE_form}, the probability $p_i$ can be calculated by the common knowledge of state-of-charge:
\begin{align}
    p_i = \int_0^{\widehat{SoC_i}} f(SoC) \rm d SoC,
\end{align}
where $\widehat{SoC_i}$ satisfies condition~\eqref{eq:BNE_form}. With the above conclusion, we have the complete expression for $h(SoC_i)$.

Then, we design an algorithm as outlined in Algorithm~\ref{alg:solve_BNE_normal}. We iterate over all EVs and solve for charging decision point $SoC_i$ for each driver. Based on the analysis in Appendix~\ref{apdx-prof_NE_form}, when $h(SoC_{d,i}) \geq 0$, the charging decision point $\widehat{SoC_i}=SoC_{d,i}$. Otherwise, it is the unique solution of $h(SoC_i)=0$ (as shown in Figure~\ref{fig:h_shape}). Since the solution of $h(SoC_i)=0$ is unique and bounded in the interval $[SoC_{d,i},SoC_{s,i}]$, we use binary search to solve that unique solution $\widehat{SoC_i}$. With all charging decision points, we obtain the charging decision of each EV driver among $N$ drivers, which is the Bayesian Nash Equilibrium.

\begin{algorithm}[t]
\small
\DontPrintSemicolon
\caption{BNE solution approach without destination distribution knowledge.}
\label{alg:solve_BNE_normal}
\KwIn{All parameters in Table~\ref{tab:parameter_symmetric} (setting 1), precision requirement $\epsilon$ (default 0.1\%).}
\For{$i=1$ \KwTo $N$}{
    Calculate $SoC_{s,i}$, $SoC_{d,i}$ using equations~\eqref{eq:dsi} and \eqref{eq:d0i}.\;
    \eIf{$h(SoC_{d,i}) < 0$}{
        Initiate $l=SoC_{d,i}$, $r=1$, $d=+\infty$.\;
        \While{$r - l > \epsilon$}{
            Calculate $m=\frac{1}{2}(l+r)$, $v_l=h(l)$, $v_r=h(r)$, $v_m = h(m)$.\;
            \eIf{$v_m =0$}{
                Break while.\;
            }{
                \eIf{$v_l \cdot v_m \leq 0$}{
                    Update $r=m$.\;
                }{
                    Update $l=m$.\;
                }
            }
        }
        $\widehat{SoC_i}=m$.\;
    }{
        $\widehat{SoC_i}=SoC_{d,i}$.\;
    }
    EV $i$ charge if $SoC_i < \widehat{SoC_i}$, not charge otherwise.\;
}
\end{algorithm}

\subsection{BNE Solving Method with Destination Distribution Knowledge}
\label{ssec:with_distribution}

In the previous subsection, we assume that EV owners use self-referencing to infer the destinations of other car owners. In this section, we liberalize this assumption and consider a situation where the EVs know destination distribution.

Given that the owner is aware of the proportion of vehicles traveling to various destinations within the system, it is essential to consider system solutions that encompass vehicles traveling to different destinations. When an EV with identical parameters travels to different destinations, the charging decision point will vary due to the change in the distance to the next charging station. We denote the decision point for EV $i$ when travelling to destination $k$ by $\widehat{SoC_i^k}$. Since all EVs in the system share the same expected queuing time, and the decision points satisfy $h(\widehat{SoC_i^k})=0$ for any $k$, we have
\begin{align} \label{eq:soc_hat_k}
\small
    \mu_q \mathbb{E}[t] = \mu_a \lambda (SoC_{S,i}^k - \widehat{SoC_i^k})^\alpha_i + V_i P_i (1-\widehat{SoC_i^k}), \quad \forall k.
\end{align}

The right side of Equation~\eqref{eq:soc_hat_k} is decreasing with $SoC_i^k$. Therefore, once one driver's charging decision point is known, the decision points for vehicles with the same parameters but different destinations can be calculated using Equation~\eqref{eq:soc_hat_k}. In consequence, their charging probabilities $p_i^k$ ($k=1,\dots,K$, different by destination) can be estimated by
\begin{align} \label{eq:p_ik}
\small
    p_i^k = \int_0^{\widehat{SoC_i^k}} f(SoC) \rm d SoC.
\end{align}

Then we re-examine the expected count of charging EVs, $\mathbb{E}[c(N)]$ in this scenario. In the $N$-player system, we denote $N_k$ as the number of EVs going to destination $k$, and denote $c(N_k)$ as the number who decide to charge among them. Since $c(N) = c(N_1) + \dots + c(N_K)$, when one EV driver traveling to destination $k^\prime$ considers the queuing cost of charging, he will consider the charging number $\mathbb{E}[c(N)]$ as one (himself) plus the sum of expectation of $k$ Binomial distributions, i.e.,
\begin{align} \label{eq:e_cn_traffic}
\small
    \mathbb{E}[c(N)] &= \sum_{k=1}^K \mathbb{E}[c(N_k)], \nonumber \\
    &= \sum_{k\neq k^\prime} (N_k \cdot p_i^k) + (N_{k^\prime}-1) \cdot p_i^{k^\prime} + 1.
\end{align}

As we assume all drivers know the diversion ratios at each intersection in the road network, the distribution of vehicles to various destinations is also known. Thus $N_k$ can be calculated with total number $N$ and distribution. With $K$ set of equations~\eqref{eq:soc_hat_k}, equation~\eqref{eq:Et} and ~\eqref{eq:e_cn_traffic}, the system, including the $k+2$ variables , is solvable.

We design the improved binary search algorithm as shown in Algorithm~\ref{alg:solve_BNE_traffic} based on the unique solution of $h(SoC_i^k)=0$. First, we assume the decision point as $SoC_i^k=m$, and calculate the decision points for those EVs with same parameters but different destinations $SoC_i^{-k}$ based on Equations~\eqref{eq:soc_hat_k}. Then, we calculate the $\mathbb{E}[t]$ using Equation~\eqref{eq:p_ik} and Equation ~\eqref{eq:e_cn_traffic}. Finally, we could determine the sign of $h(SoC_i^k)$ and adopt the binary search.

\begin{algorithm}[t]
\small
\DontPrintSemicolon
\caption{BNE solution approach with destination distribution knowledge.}
\label{alg:solve_BNE_traffic}
\KwIn{All parameters in Table~\ref{tab:parameter_symmetric} (setting 2), precision requirement $\epsilon$ (default 0.1\%).}
\For{$i=1$ \KwTo $N$}{
    Initiate $l=SoC_{d,1}$, $r=1$.\;
    \While{$r - l > \epsilon$}{
        Assign $m=\frac{1}{2}(l+r)$.\;
        Calculate the charging probabilities using equations~\eqref{eq:soc_hat_k} and equation~\eqref{eq:p_ik}.\;
        Calculate the $\mathbb{E}[t]$ using equation~\eqref{eq:Et} and ~\eqref{eq:e_cn_traffic}.\;
        Calculate $h(m)$ using $\mathbb{E}[t]$.\;
        Calculate $h(l)$ and $h(r)$ similarly.\;
        \eIf{$h(m)=0$}{
            Break while.\;
        }{
            \eIf{$h(l)\cdot h(m) \leq 0$}{
                Update $r=m$.\;
            }{
                Update $l=m$.\;
            }
        }
    }
    EV $i$ charge if $SoC_i < \widehat{SoC_i}$, not charge otherwise.\;
}
\end{algorithm}

\section{Discussion: Results with diverse parameters}
\label{sec:diverse_para}
In this section, we discuss EVs with different weighting parameters and diverse target SoC values. We denote the heterogeneous parameters as $\mu_{a,i}$, $\mu_{q,i}$ and $SoC_{t,i}$. Under heterogeneous setting, the payoff in Equation~\eqref{eq:utility} is

\begin{equation}
\small
\label{eq:utility_div}
\begin{aligned}
  \Pi_i(\boldsymbol{x}) &= \Pi_{a,i}+\Pi_{c,i}+\Pi_{q,i} \\
  &=\left\{
  \begin{aligned}
    & -V_i (\textcolor{black}{SoC_{t,i}}-{SoC}_i) \cdot P -\mu_{q,i} \cdot t, & \text{if $x_i=1$}, \\
    & -\mu_{a,i} A_i, & \text{if $x_i=0$}.
  \end{aligned}
  \right.
\end{aligned}
\end{equation}

The accessibility of this information is also worth discussing. Similar to the State-of-charge, these parameters are also the \textit{private information} known only to the driver. How an EV driver evaluates the parameters of other drivers will directly influence the results. In this paper, we consider two information accessibility scenarios:
\begin{itemize}
    \item \textbf{Self-referencing}: The driver estimate other EVs' parameters by self-referencing \cite{burnkrant1995effects}, i.e., the driver assumes other EVs have the same parameters as he do.
    \item \textbf{Expectation}: We assume the distribution of these parameters is common knowledge known to all drivers.
\end{itemize}

First, we consider the self-referencing scenario. Under self-referencing, function $h(SoC_i)$ is
\begin{align} \label{eq:h_diverse_sr}
\small
    h(SoC_i) :=& \mathbb{E}[\Pi(x_i=0)|SoC_i] - \mathbb{E}[\Pi(x_i=1)|SoC_i], \nonumber \\
    =& \mathbb{E}[-\mu_{a,i}A_i] - \mathbb{E}[-V_iP(SoC_{t,i}-SoC_i) -\mu_{q,i} \cdot t], \nonumber \\
    =& -\mu_{a,i} A_i + V_iP(SoC_{t,i}-SoC_i) + \mu_{q,i} \mathbb{E}[t]. \nonumber
\end{align}
Taking the derivation, we still have  $\frac{\rm d h(SoC_i)}{\rm d SoC_i} > \frac{\rm d h(SoC_{s,i})}{\rm d SoC_i} = -V_iP $. There still exists only one point satisfying condition~\eqref{eq:BNE_form_h}, thus the existence and uniqueness results still hold. The uniqueness-based design of Algorithms~\ref{alg:solve_BNE_normal} and ~\ref{alg:solve_BNE_traffic} are also effective in solving the BNE in this scenario.

Next, we consider the Expectation scenario. When EV $i$ evaluating the decision of any other EV $j$, he will estimate the expected payoff of EV $j$ as
\begin{align}
\small
    h(SoC_j) :=& \mathbb{E}[\Pi(x_j=0)|SoC_j] - \mathbb{E}[\Pi(x_j=1)|SoC_j], \nonumber \\
    =& \mathbb{E}[-\mu_{a,j}A_j] - \mathbb{E}[-V_jP(SoC_{t,j}-SoC_j) -\mu_{q,j} \cdot t], \nonumber \\
    =& -\mathbb{E}[\mu_a] A_i + V_iP(\mathbb{E}[SoC_t]-SoC_i) + \mathbb{E}[\mu_q] \mathbb{E}[t], \nonumber
\end{align}
where $\mathbb{E}[\mu_a]$, $\mathbb{E}[\mu_q]$ and $\mathbb{E}[SoC_t]$ are the expected values of the parameters. The parameters are all known constant thus the proposed existence and uniqueness theory hold, and the algorithms are effective in solving the decision point. \footnote{\textcolor{black}{Please note that this decision point might deviate from the result where EV $j$ values himself, since the parameters are private information only known to EV $j$, while EV $i$ estimates them using statistical distribution.}} With the estimated charging decision of other EVs, EV $i$ makes its own charging decisions by comparing the expectations of charging and not charging.

Overall, our results remain valid with different parameters. Operators can make more accurate predictions of vehicle owners' charging behavior by considering different users' parameter preferences.

\section{Benefits of Accurate Prediction}
\label{sec:benefit_cs}

An accurate predictions of charging behavior would enhance many of the economic benefits of the road network. In this section, we give an example of enhancing the economic benefits of charging stations.

Let's consider the scenario where a highway charging station with energy storage needs to determine the amount of power to be purchased from the grid each hour to reduce their operating costs. Since the price of electricity has peak and valley volatility, the charging station wants to try to avoid buying electricity during peak periods while ensuring enough power to provide service. We denote the price of electricity by $\mathbf{Pe} = [Pe_1,\dots,Pe_T]^\top$ and the charging station's power purchase behavior by $\mathbf{X} = [X_1,\dots,X_T]^\top$. We use $\mathbf{D} = [D_1,\dots,D_T]$ to denote the charging station's forecast of charging demand. The charging station wants to minimize the charging cost, that is he has to solve the following optimization problem:
\begin{subequations} \label{eq:CS_optimization}
\small
\begin{align}
&\min_{\mathbf{X}} &\mathbf{X}^\intercal \mathbf{Pe},\\
&s.t. &\mathbf{0} \leq \mathbf{X} \leq \overline{\mathbf{X}}, \label{seq:CS_optimization_bound_contraint}\\
&     &E_t = E_0 + \sum_{\tau=1}^t X_{\tau} - \sum_{\tau=1}^{t-1} D_t,\\
&     &0 \leq E_t \leq \overline{E}, \label{seq:CS_optimization_energy_constraint}
\end{align}
\end{subequations}
where $\overline{\mathbf{X}}$ denotes the maximum charging capacity at time $t$, $E_t$ denotes the energy in the storage, and the storage level should be between zero and maximum capacity $\overline{E}$. \color{black} The optimization problem \eqref{eq:CS_optimization} can be reduced to a linear optimization problem, and there are many sophisticated toolboxes for solving it, such as \textit{linprog} in Matlab. With the optimal solution of \eqref{eq:CS_optimization}, the charging station operator minimizes the charging cost and therefore improves the profit. \color{black} From this optimization problem we can see that the charging station's forecast of charging demand for electric vehicles $\mathbf{D}$ speaks to his revenue. We will show in an experimental in Section~\ref{ssec:benefits_cs} to show that compared to other methods, our method enhances the revenue of the charging station.

\section{Experiments}
\label{sec:experiment}

In the experimental section, we employ numerical simulations to validate the efficacy of our proposed model. We first used a high-speed road network to simulate the system under different risk averse preferences for parameters such as road network exit departure rates, charging station queue lengths and charging demand, and average queuing times for electric vehicles. Second, we use the model in Section~\ref{sec:benefit_cs} to explore the benefits of highway charging stations under this model compared to other bounded rationality models.

\subsection{Parameter Settings}

We consider a scenario involving the continuous 24-hour highway arrival flow and a highway network as shown in Figure~\ref{fig:traffic_network}. The primary parameters are detailed in Table~\ref{tab:parameter_symmetric}. We use data from highways England A1, A2 highway arrival in January 2023 respectively to simulate the arrival rate of Entrance 2 and Entrance 1 \cite{highwaysengland}. We assume that EVs constitute 2\% of the total vehicle population \cite{EVpenetration}. Then, we designate each 15-minute interval as a system to compute the charging decisions of each EV based on the queuing status in the previous interval. We initialize the queuing time of preceding EVs $t_0=0$ at the commencement of our simulation (0:00). The CS parameters are shown in Table~\ref{tab:parameter_symmetric}. As for the number of charging piles and charging price, we set the $P=2.3$ CNY/kWh and $k=4$ for CS 1, CS 2 and CS 3, $P=2.5$ CNY/kWh and $k=6$ for CS 4 and CS 5. \color{black}
As for information accessibility, we calculate results in both settings in single CS simulations.
In the road network simulations, we adopt the scenario where EVs know the destination distribution information (Setting 2).
\color{black}

\begin{figure*}[t]
\centering
\begin{subfigure}[t]{0.5\textwidth}
    \centering
    \includegraphics[width=0.9\linewidth]{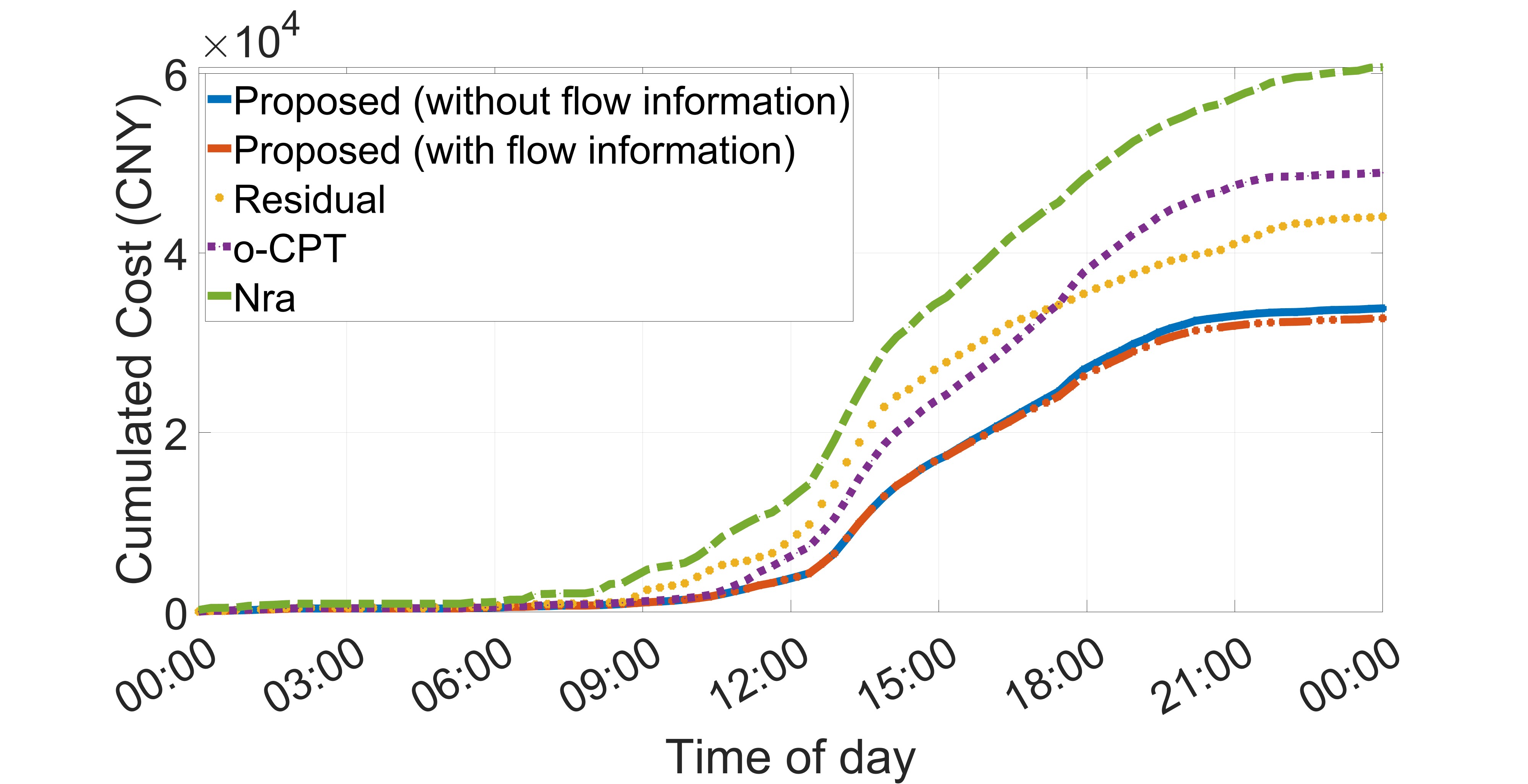}
    \caption{Weekdays.}
    \label{fig:cost_EV_weekdays}
\end{subfigure}%
\hfill
\begin{subfigure}[t]{0.5\textwidth}
    \centering
  \includegraphics[width=0.9\linewidth]{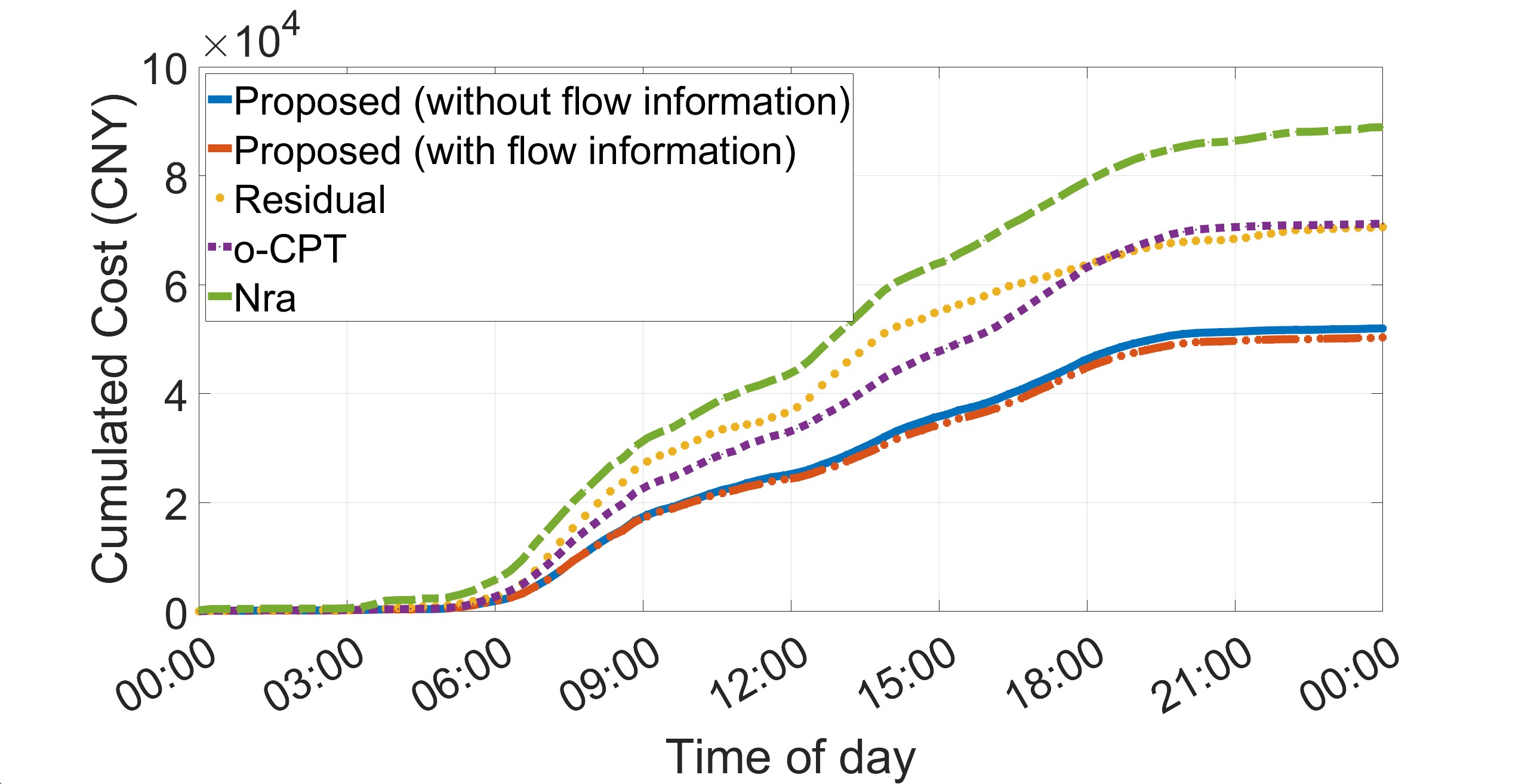}
  \caption{Weekends.}
  \label{fig:cost_EV_weekends}
\end{subfigure}%
\caption{\textcolor{black}{Cumulative social cost for one station.}}
\label{fig:cost_EV}
\end{figure*}

\begin{figure}[t]
    \centering
    \includegraphics[width=0.9\linewidth]{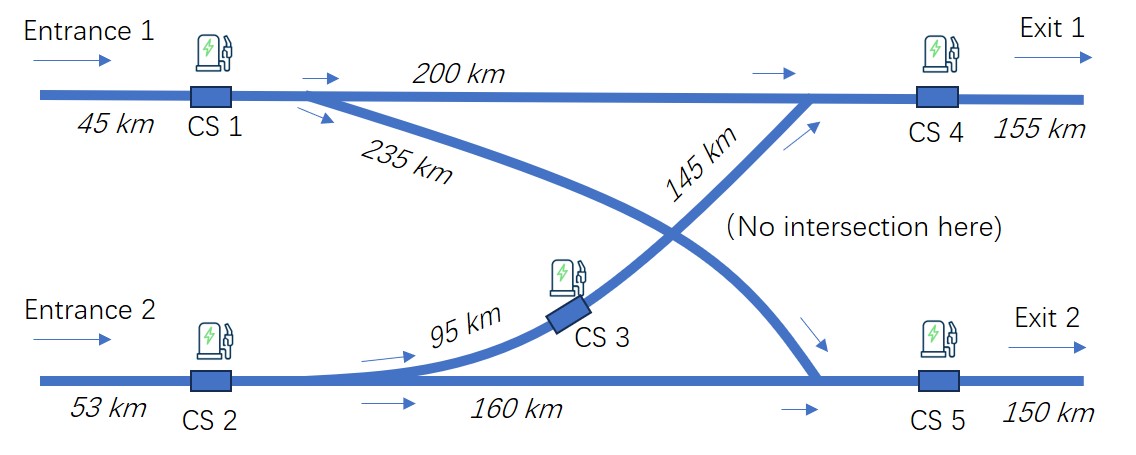}
    \caption{Simulation of the transportation network with distances between CSs: CS 1 to CS 4 is 200 km, CS 1 to CS 5 is 235 km, CS 2 to CS 3 is 95 km, CS 2 to CS 5 is 160 km, and CS 3 to CS 4 is 145 km.}
    \label{fig:traffic_network}
\end{figure}

\begin{table}
    \centering
    \begin{tabular}{cccc}
        \toprule
        Parameters & Values & Parameters  & Values \\
        \midrule
        $\alpha_i$ & 2.0 & $\lambda_i$ & 1.5 \\
        $d_0$ & $\mathcal{N}(600km,50km)$ & $d_s$ & $\mathcal{N}(400km,50km)$\\
        $V_i$ & [40,60,80,100] kWh & $soc$ &  $\mathcal{U}(20\%,100\%)$\\
        $\mu_q$ & 30 CNY/hour & $\mu_a$ & 3000 CNY \\
        $SoC_t$ & 100\% & $f(SoC)$ & uniformly in [0,1] \\
        $E_0$ & 0 kWh & $\overline{E}$ &  7 MWh \\
        $\mathbf{\overline{X}}$ & $\mathbf{4}$ Mw & $\mathbf{Pe}$ &
        \cite{SingaporeElectricityPrice} \\
        \bottomrule
    \end{tabular}
    \caption{Parameters settings.}
    \label{tab:parameter_symmetric}
\end{table}

As for EV parameters, we generate these parameters from distributions to account for the variability. In particular, the battery capacity $V_i$ is randomly selected from the set \{40 kWh, 60 kWh, 80 kWh, 100 kWh\}, while the maximum and minimum travel distances when fully charged ($d_0$, $d_s$) are drawn from normal distributions $\mathcal{N}(600km,50km)$ and $\mathcal{N}(400km,50km)$, respectively. The SoC is drawn from uniform distribution $\mathcal{U}(20\%,100\%)$. For weighting parameters $\mu_a$ and $\mu_q$, we pick $\mu_q=30$ CNY/hour and a large $\mu_a = 3000$ CNY because of the inequality~\eqref{eq:mu_a_much_greater}. Given our chosen parameters, the left-hand side of inequality~\eqref{eq:mu_a_much_greater} is approximately tenfold the right-hand side.

\subsection{EV Drivers' Cost}
\label{ssec:EV_cost}

In this experiment, we verify the cumulative total EV cost (which also reflects the negative social welfare) under different prediction methods. We simulate the EV flow from Entrance 1 and a charging station (CS) with 8 charging piles and a service price of 1.15 CNY/kWh, using our proposed methods as well as benchmark decision-making methods. The selected benchmarks are as follows:

\begin{itemize}
    \item \textbf{Residual}: A common method to leave some constant redundancy for range anxiety \cite{yang2016predicting,hu2018analyzing} to capture the bounded rationality. It assumes that charging is required when the state-of-charge falls below the energy need plus the redundancy.
    \item \textbf{o-CPT}: A bounded rational model proposed in \cite{hu2019modeling}. It captures the risk-aversion effects but does not consider the congestion caused by multiple vehicles charging at the same time.
    \item \textbf{Nra}: The models without considering bounded rationality when analyzing the EV charging behaviors \cite{yin2024research,williams2024driving}.
\end{itemize}

We then calculated the cumulative total EV cost over time, as shown in Figure~\ref{fig:cost_EV}. The social cost under our proposed methods is consistently lower than that of the benchmarks. For example, the average cumulative total cost is 33,836 CNY on weekdays with our method without destination information on the EV flow, which is only 56\% of the Nra method (60,689 CNY) and 77\% of the Residual method (44,023 CNY). Our method also reduced the social cost based on weekend data, as shown in Figure~\ref{fig:cost_EV_weekends}. This indicates that our approach can better predict the actual charging behavior of EVs, as car owners will always opt for decisions that minimize their own costs.

Additionally, we find that the cumulative cost with destination distribution information in EV flow is consistently better than without this information. This indicates that EV owners can make more informed decisions when they have access to more information. For highway operators, better predictions can be made by using questionnaires to investigate whether owners are familiar with the distribution of destinations in the traffic flow.

\subsection{Rationality Effects on EV Charging Decision}

In this experiment, we illustrate the rational sensitive on the charging decision.
To better show the impact on EV charging decisions, we remove EV randomness and simulate a symmetric 5-EV scenario. The EV parameters are the same in Table~\ref{tab:parameter_symmetric}, expect we removed randomness and set $d_0 = 600 km$, $d_s = 400 km$, $V_i = 80 kWh$. The queuing time for previous EVs $t_0$ is set to 2.5 hours.

We vary the risk aversion tendency (rationality parameter) $\lambda$ from 1 to 5, and calculate the charging decision points $\hat{SoC}$. The simulation result is shown in Figure~\ref{fig:charging_decision_points_ra}. We also show the case of perfect rational as benchmark in Figure~\ref{fig:charging_decision_points_nra}.

\begin{figure*}[t]
    \centering
    \begin{subfigure}{\columnwidth}
        \centering
        \includegraphics[width=0.9\linewidth]{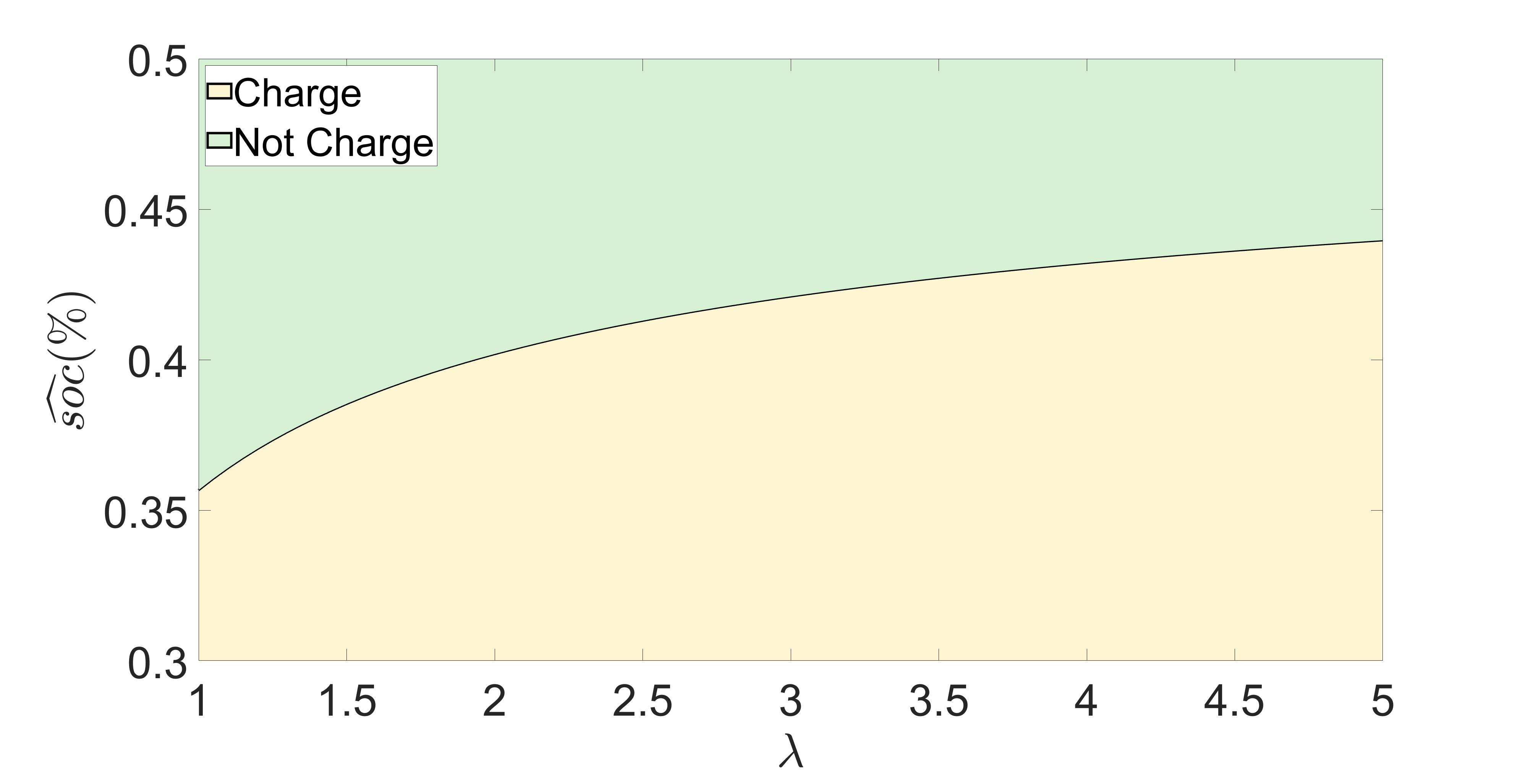}
        \caption{With bounded rationality.}
        \label{fig:charging_decision_points_ra}
    \end{subfigure}
    \hfill
    \begin{subfigure}{\columnwidth}
        \centering
        \includegraphics[width=0.9\linewidth]{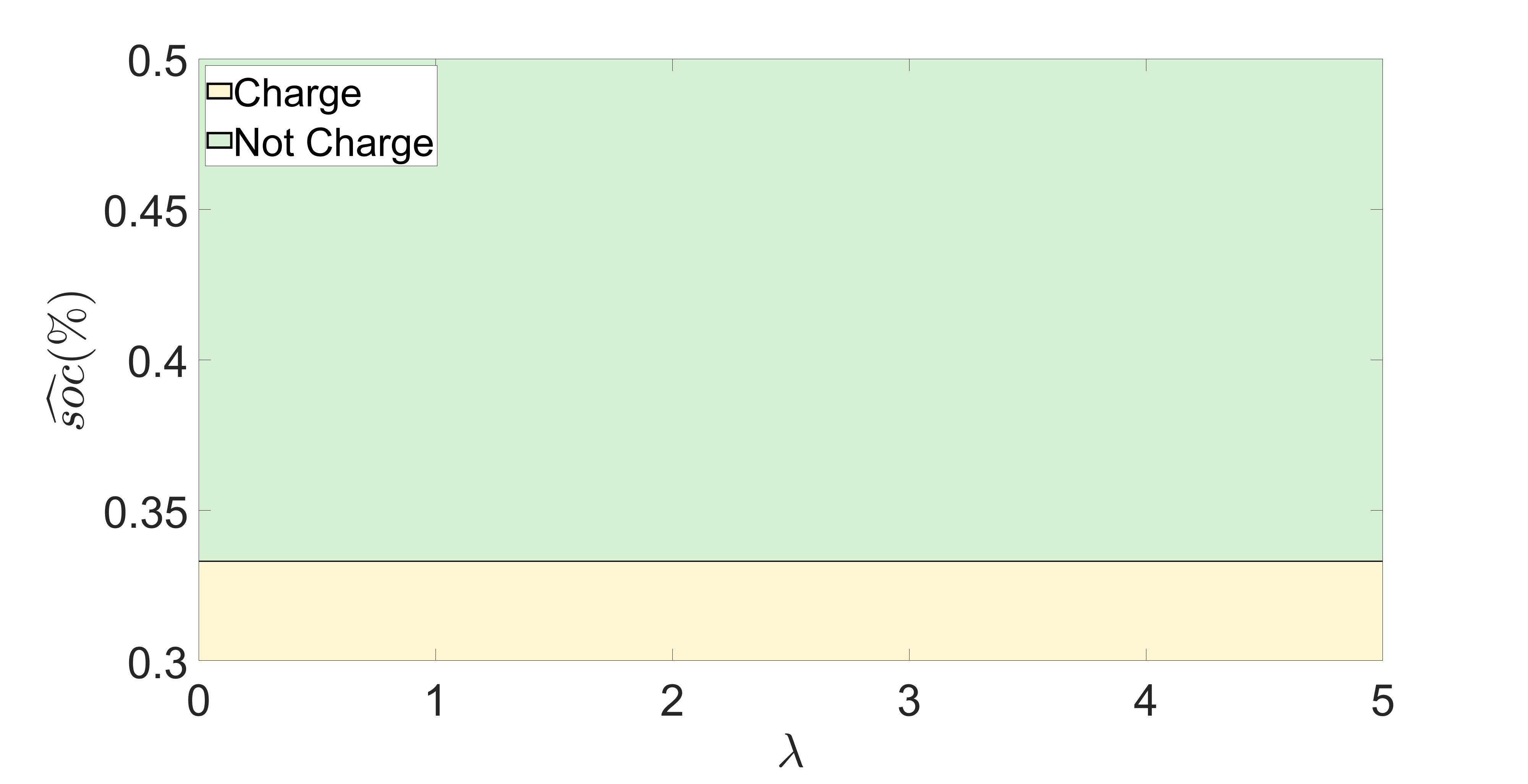}
        \caption{Without bounded rationality.}
        \label{fig:charging_decision_points_nra}
    \end{subfigure}
    \caption{Charging decisions as rationality varies.}
    \label{fig:charging_decision_points}
\end{figure*}

\begin{figure*}[t]
\centering
\begin{subfigure}[t]{0.33\textwidth}
  \includegraphics[width=0.9\linewidth]{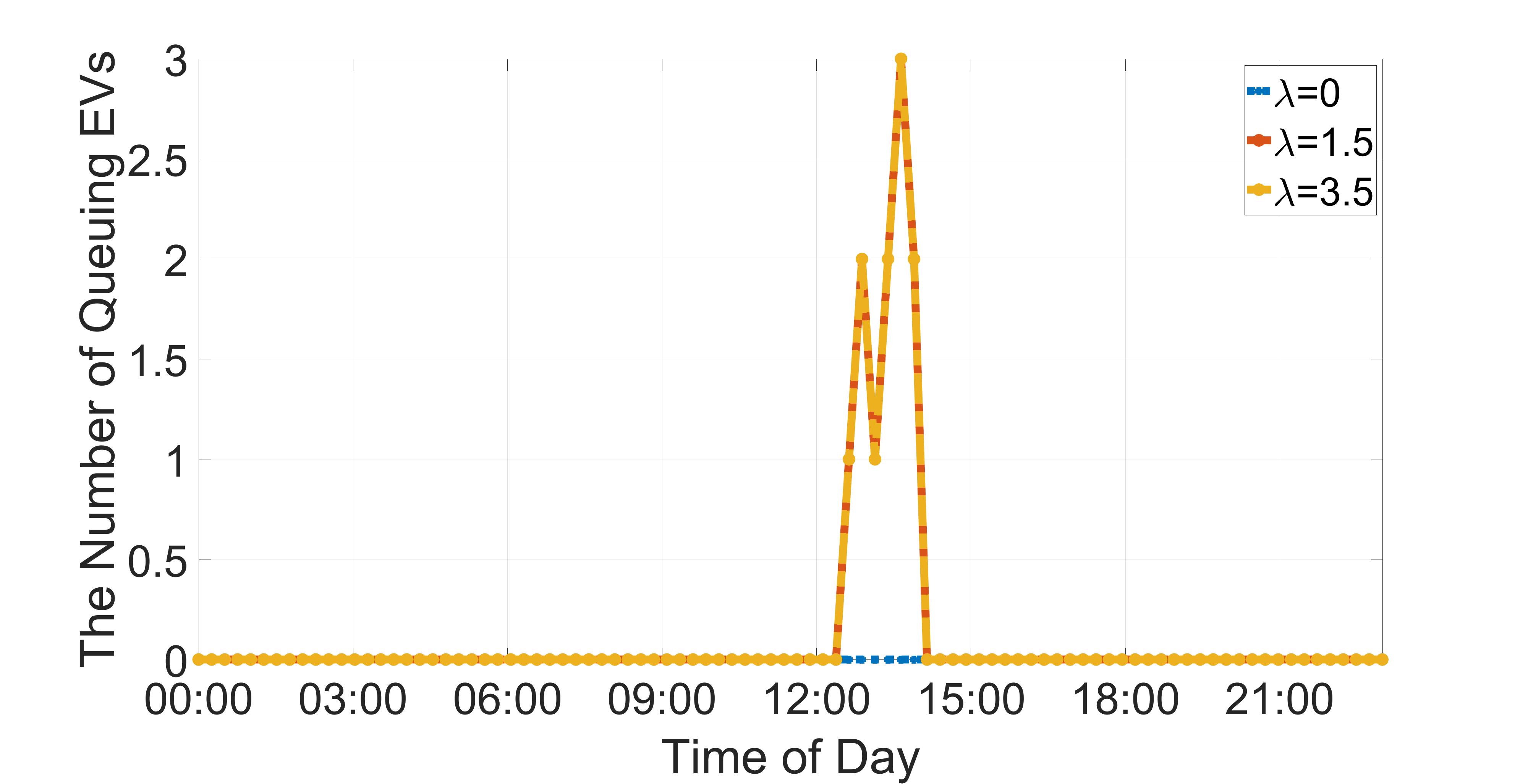}
  \caption{CS 1}
  \label{fig:queuing_num_cs1}
\end{subfigure}%
\hfill
\begin{subfigure}[t]{0.33\textwidth}
  \includegraphics[width=0.9\linewidth]{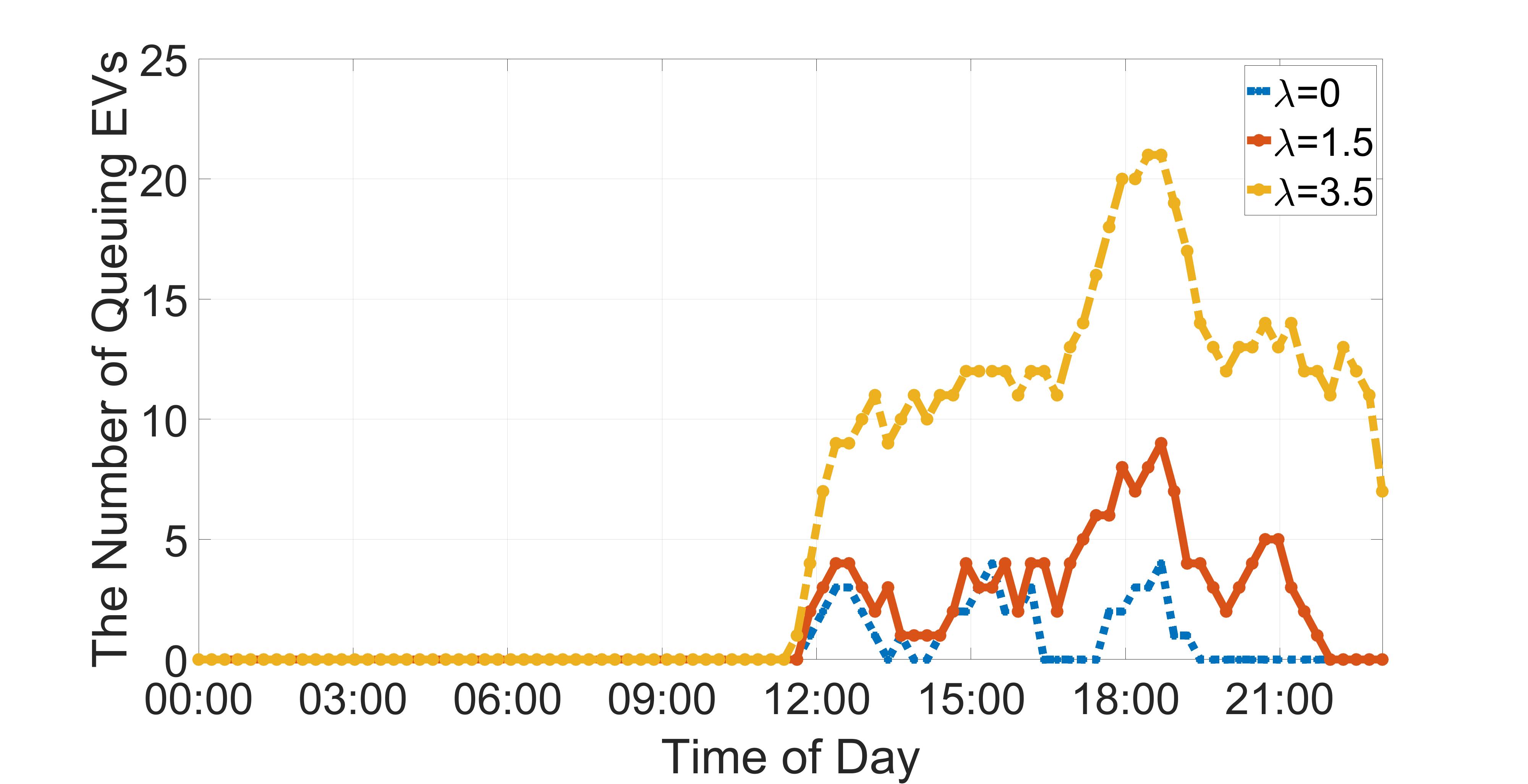}
  \caption{CS 2}
  \label{fig:queuing_num_cs2}
\end{subfigure}%
\hfill
\begin{subfigure}[t]{0.33\textwidth}
  \includegraphics[width=0.9\linewidth]{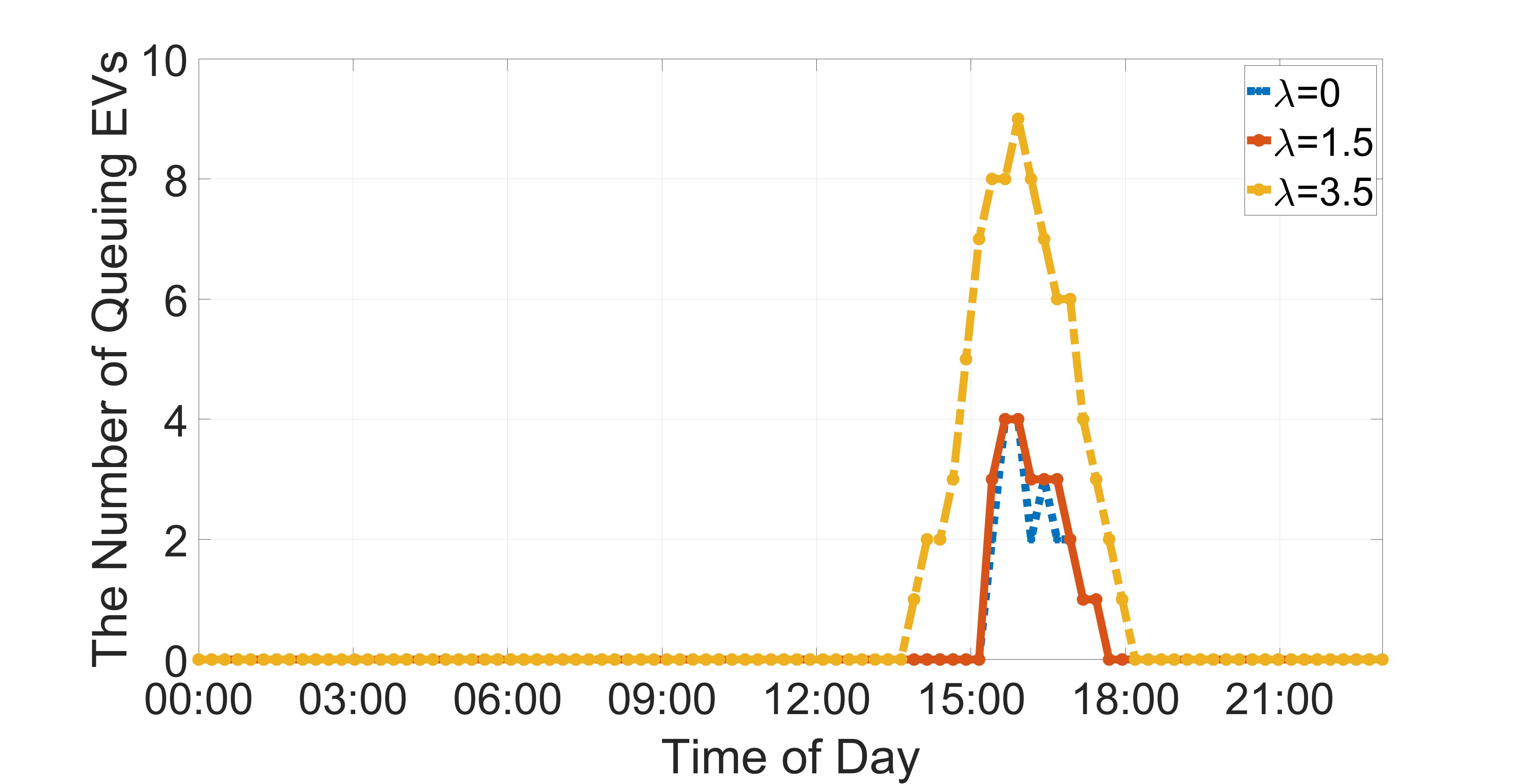}
  \caption{CS 3}
  \label{fig:queuing_num_cs3}
\end{subfigure}%

\hspace*{\fill}
\begin{subfigure}[t]{0.33\textwidth}
  \includegraphics[width=0.9\linewidth]{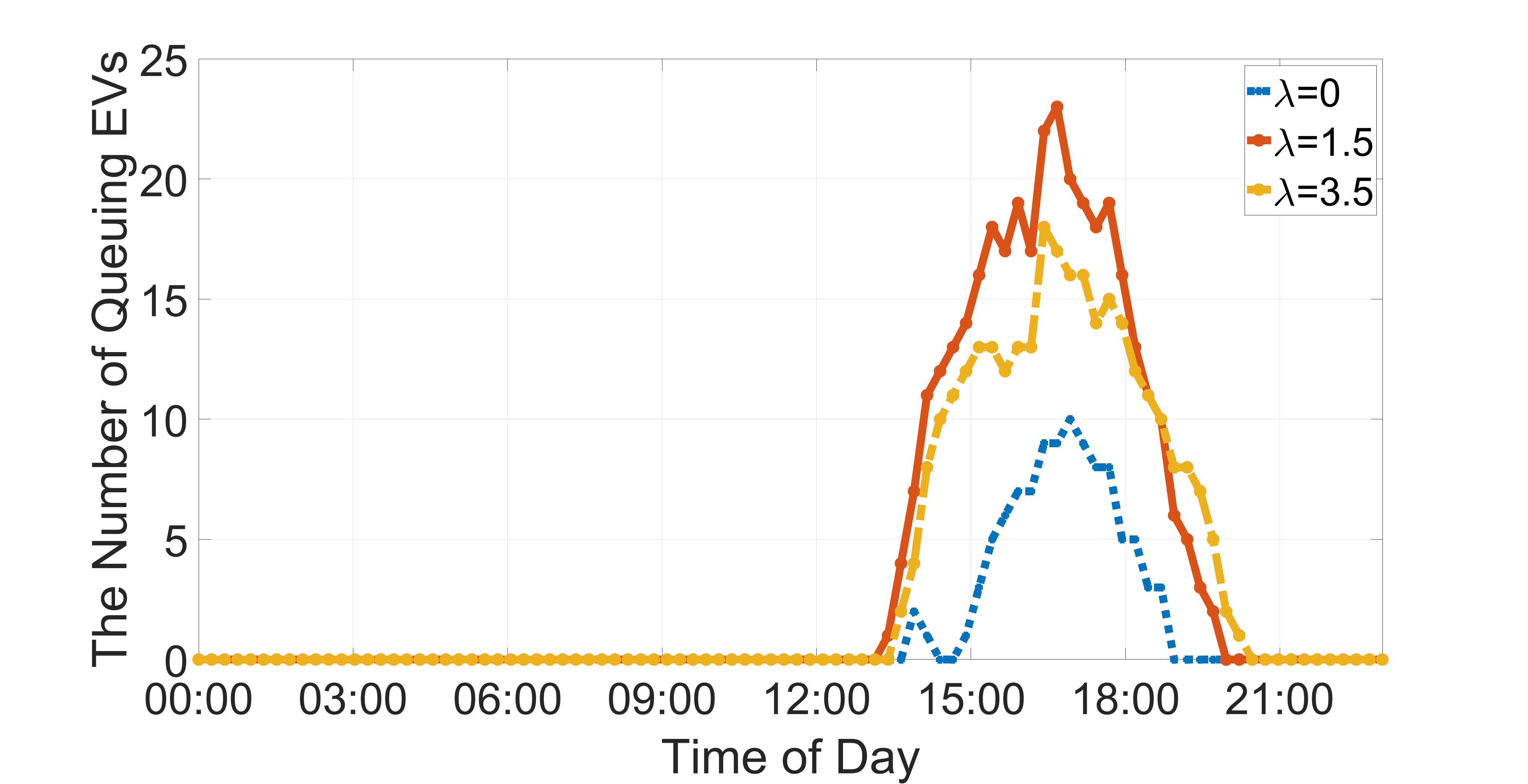}
  \caption{CS 4}
  \label{fig:queuing_num_cs4}
\end{subfigure}%
\hfill
\begin{subfigure}[t]{0.33\textwidth}
  \includegraphics[width=0.9\linewidth]{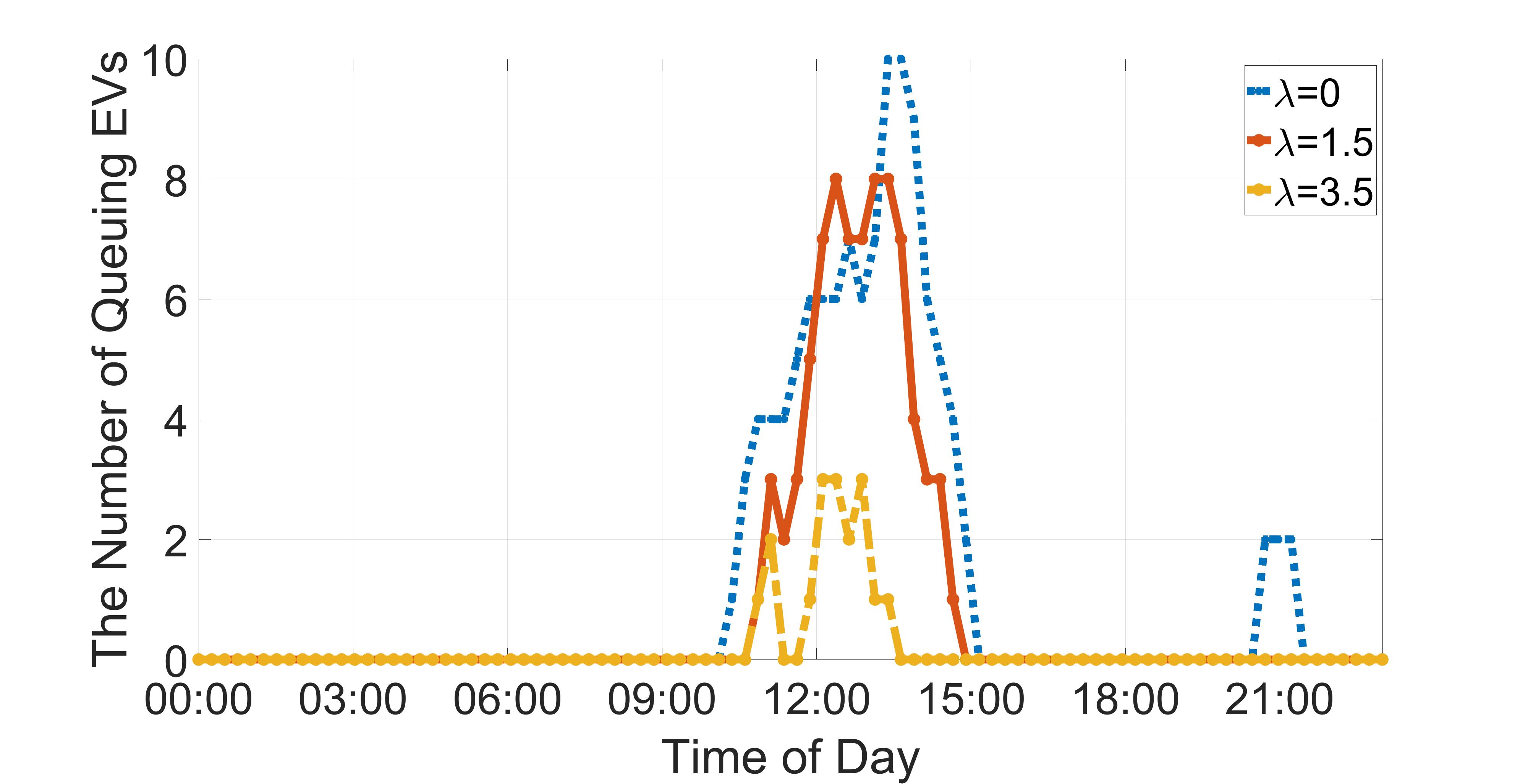}
  \caption{CS 5}
  \label{fig:queuing_num_cs5}
\end{subfigure}%
\hspace*{\fill}
\caption{Queuing length as rationality varies.}
\label{fig:queuing_num}
\end{figure*}

As depicted in Figure~\ref{fig:charging_decision_points}, we observe that the charging decision points differ remarkably by different risk aversion tendency. For example, EV driver will charge if his state-of-charge is lower than 41.8\% when $\lambda=2$, but the decision point is $33.3\%$ if we ignore the range anxiety. Besides, we observe the charging decision point are significantly influenced by drivers' risk aversion tendencies. For example, the decision point increases from 41.8\% to 45.7\% as $\lambda$ increases from 2 to 5. This result illustrates that range aversion tendency has a significant impact on the charging decision among EV drivers on the highway.

\subsection{Rationality Effects on Charging Station Parameters}
\subsubsection{Queuing Length}

We show in Figure~\ref{fig:queuing_num} the queues at each charging station with different range anxiety of EVs $\lambda$. As shown in Figure 1, the increasing anxiety of EV owners significantly affects the queue length at the charging stations. For example, in CS 2 (Figure~\ref{fig:queuing_num_cs2}), the queue length at 18:45 is 21 when $\lambda=3.5$, which is nearly double the length when $\lambda=1.5$ (9) and four times the length when perfect rational ($\lambda=0$) (4). A similar trend is observed in CS 3 (Figure~\ref{fig:queuing_num_cs3}), but CS 4 and CS 5 exhibit slightly different patterns. In CS 4, the queue length under $\lambda=1.5$ exceeds that under $\lambda=3.5$ between 12:00 and 18:30. In CS 5, the queue length under $\lambda=0$ is higher than that under $\lambda=3.5$ from 10:15 to 15:00. This discrepancy arises because CS 4 and CS 5 are located farther from the entrance in the road network, leading more EV owners to opt for charging at earlier stations when $\lambda$ increases, thereby reducing the number of vehicles charging at CS 4 and CS 5.

\subsubsection{Charging Demand}

\begin{figure*}[t]
\centering
\begin{subfigure}[t]{0.33\textwidth}
  \includegraphics[width=\linewidth]{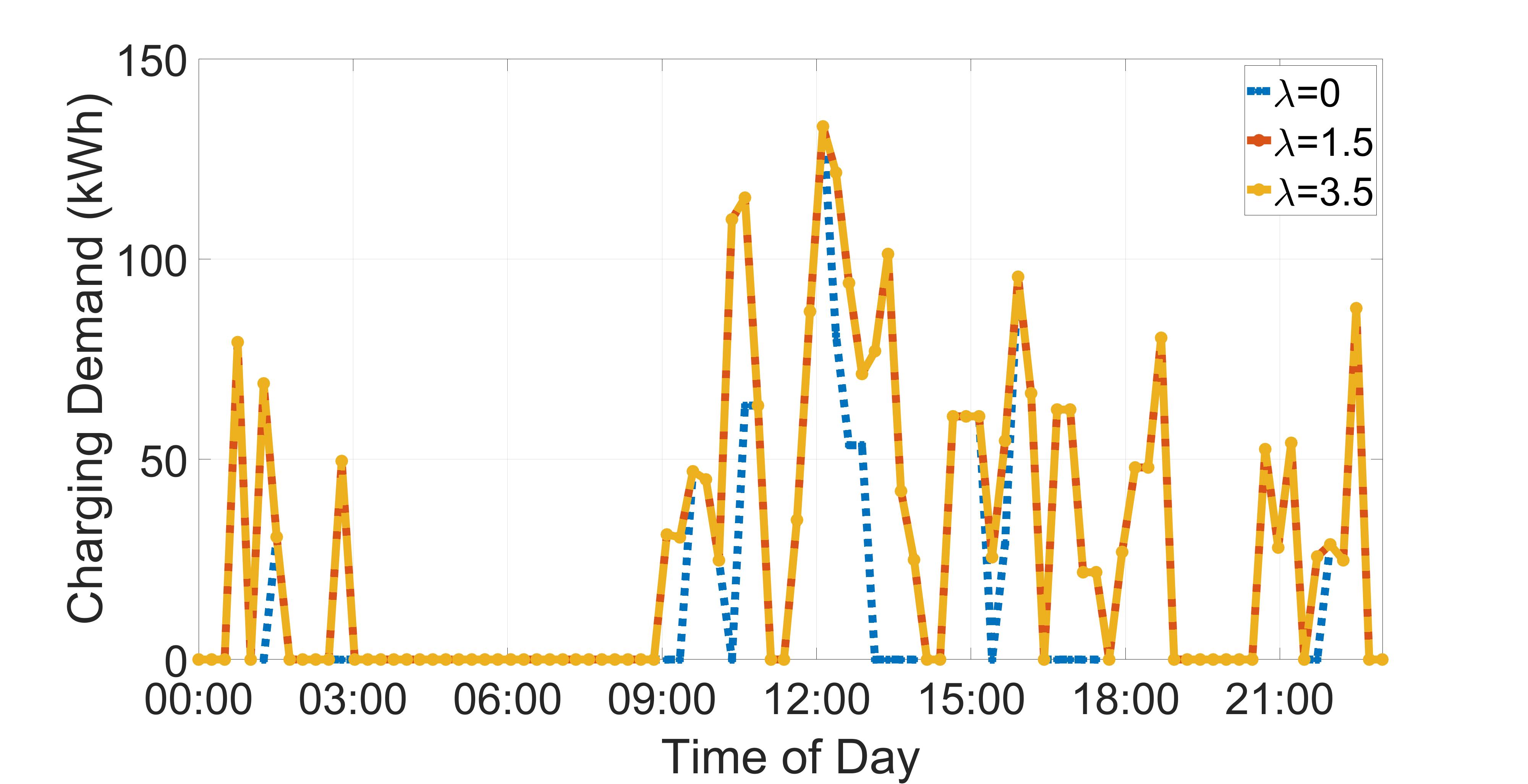}
  \caption{CS 1}
  \label{fig:charging_demand_cs1}
\end{subfigure}%
\hfill
\begin{subfigure}[t]{0.33\textwidth}
  \includegraphics[width=\linewidth]{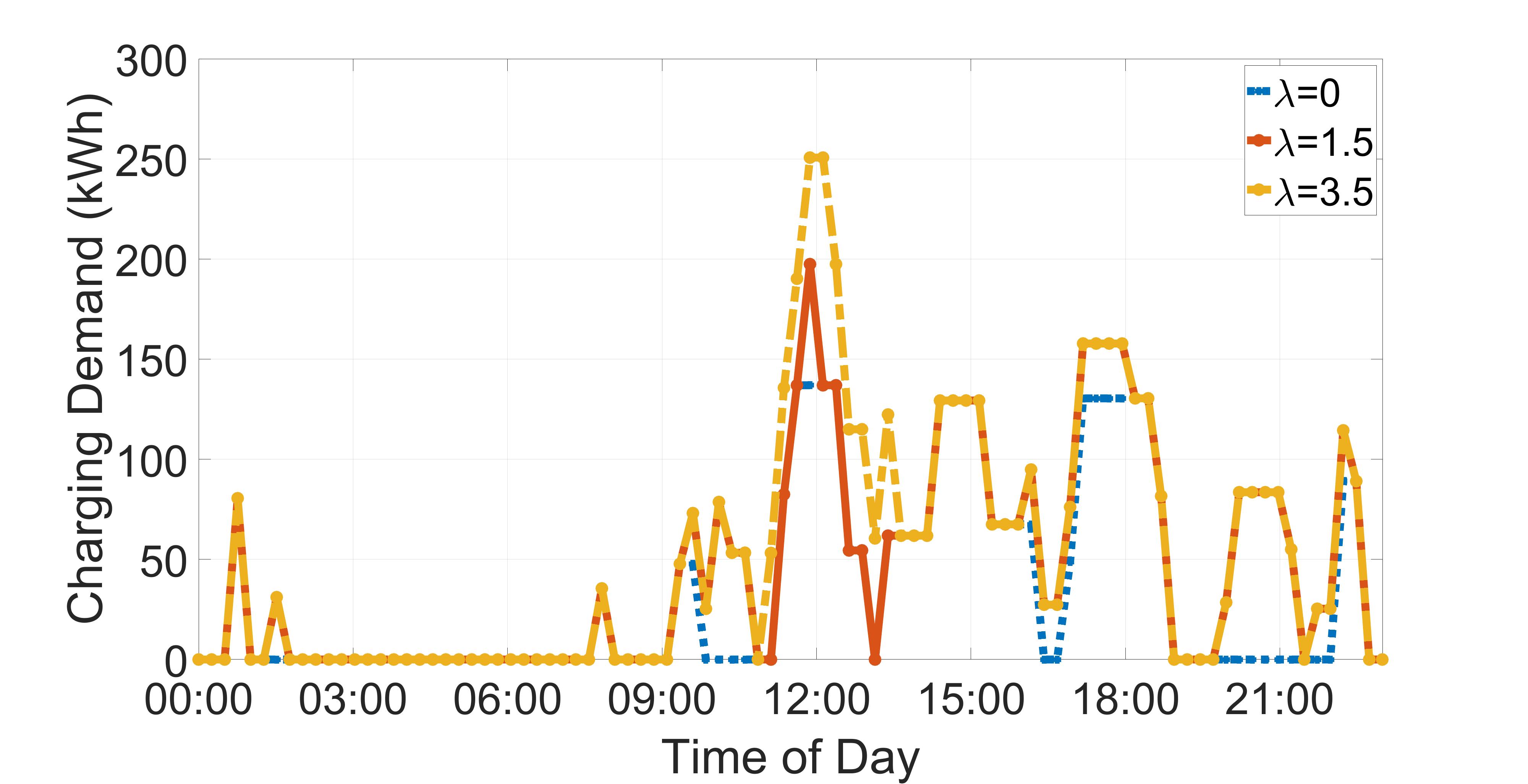}
  \caption{CS 2}
  \label{fig:charging_demand_cs2}
\end{subfigure}%
\hfill
\begin{subfigure}[t]{0.33\textwidth}
  \includegraphics[width=0.9\linewidth]{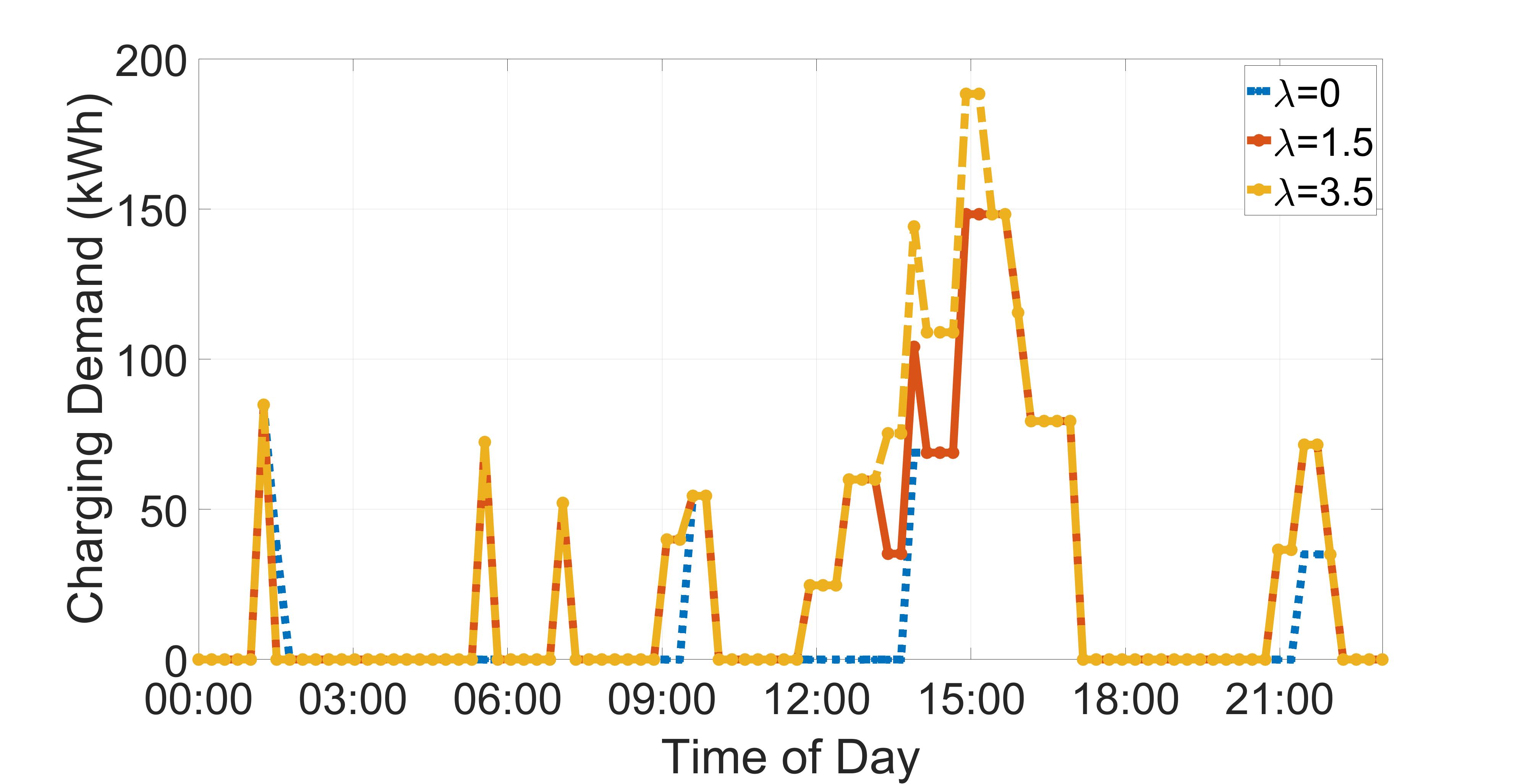}
  \caption{CS 3}
  \label{fig:charging_demand_cs3}
\end{subfigure}%

\hspace*{\fill}
\begin{subfigure}[t]{0.33\textwidth}
  \includegraphics[width=0.9\linewidth]{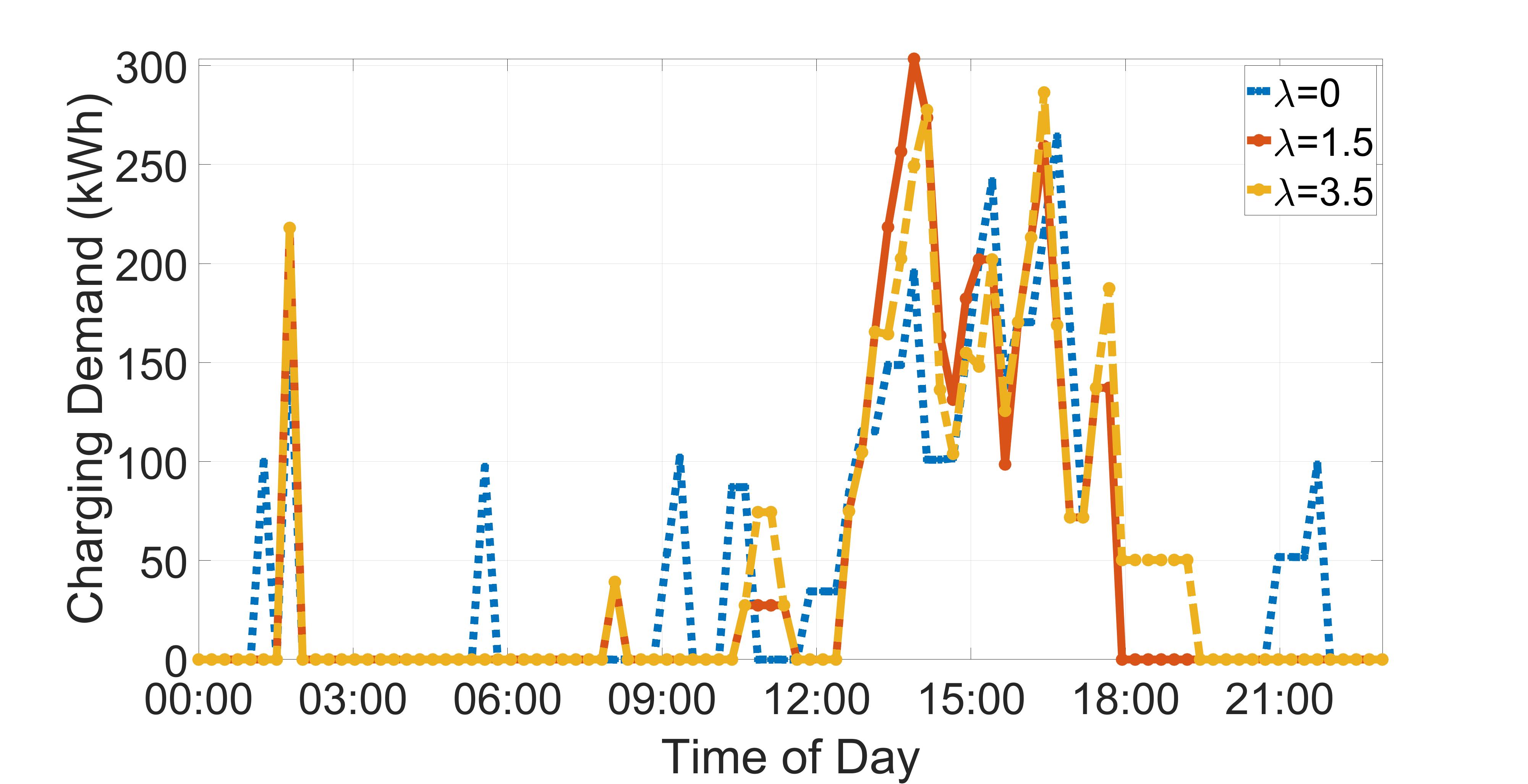}
  \caption{CS 4}
  \label{fig:charging_demand_cs4}
\end{subfigure}%
\hfill
\begin{subfigure}[t]{0.33\textwidth}
  \includegraphics[width=0.9\linewidth]{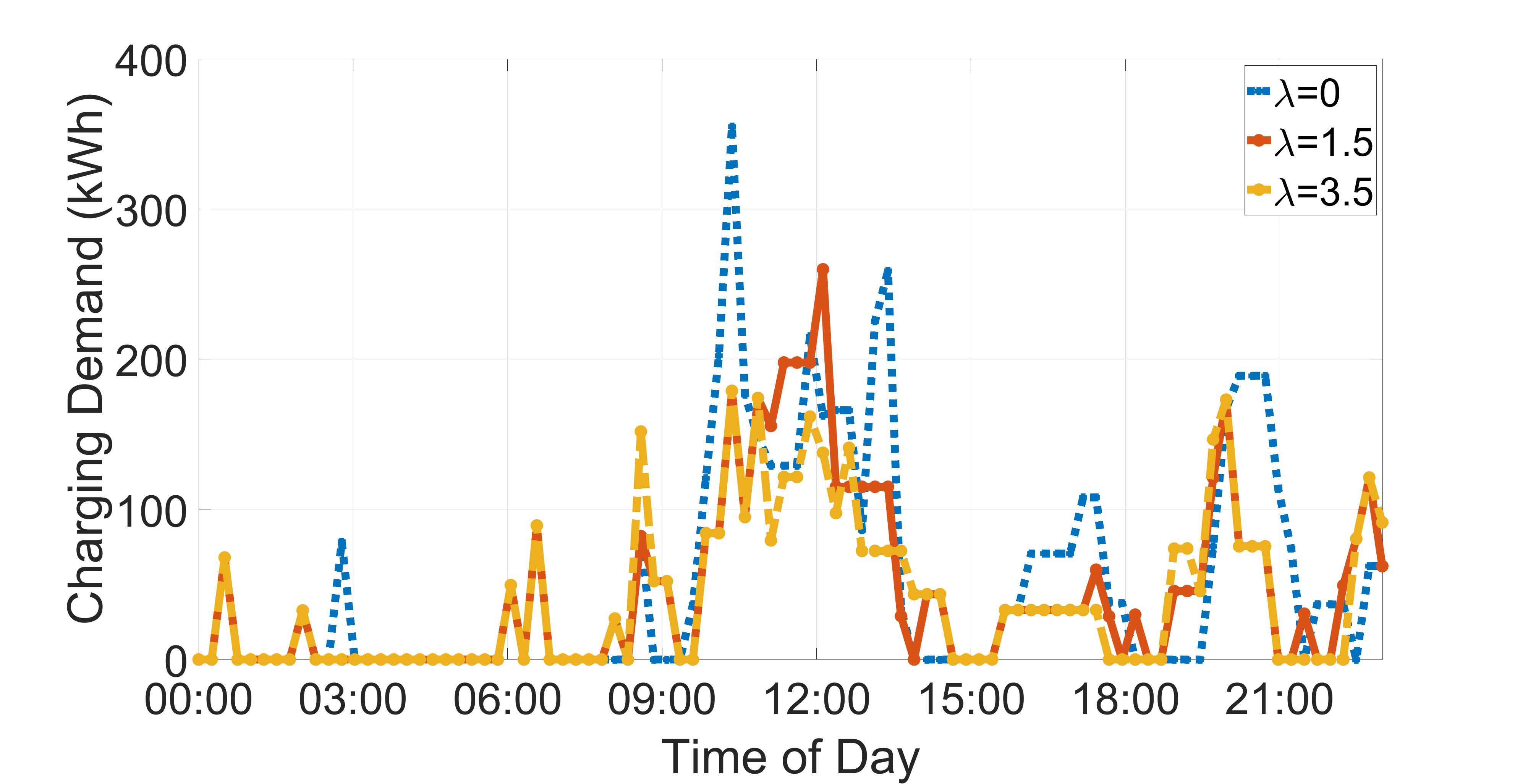}
  \caption{CS 5}
  \label{fig:charging_demand_cs5}
\end{subfigure}%
\hspace*{\fill}
\caption{Charging demand as rationality varies.}
\label{fig:charging_demand}
\end{figure*}

Figure~\ref{fig:charging_demand} illustrates the variation in charging demand at charging stations within the road network under different risk aversion preferences. As depicted, the risk aversion preferences of EV owners significantly impact the stations. For instance, at CS 2, the charging demand between 11:45 and 12:00 is approximately 140 kWh when perfect rational ($\lambda=0$), increasing to around 200 kWh when $\lambda=1.5$, and 250 kWh when $\lambda=3.5$. For charging stations located farther from the entrance, the charging demand may decrease as aversion preference increases. For example, at CS 5, the charging demand is approximately 350 kWh when $\lambda=0$, but decreases to around 180 kWh when $\lambda=3.5$. This difference arises because EV owners with higher risk aversion tend to prefer charging at the front stations (CS 1, CS 2, and CS 3), thereby reducing the charging demand at CS 4 and CS 5.

Interestingly, we observe the magnitude of change in charging demand is not as pronounced as in queuing status. For example, at CS 2, the ratio of queue lengths can be as high as 4:2:1 in the most significant hours (18:45), but the ratio of charging demand in the most significant hours is only 25:25:14 (11:45). This is because some of the EVs would have liked to charge one more time, but they decided to leave when they saw the long EV queue.

\subsection{Rationality Effects on Transportation Network Parameters}

We show the departure rates of road network exits under different risk averse preferences of EV owners.The departure rates for Exit 1, 2 are shown in Figure~\ref{fig:departure_rate_exit1}, Figure~\ref{fig:departure_rate_exit2}, respectively.

As illustrated in the figures, the duration of peak hours for the departure rate extends with the increase in owners’ range anxiety tendency. Specifically, the period with more than 10 vehicles extends from 12 to 14 hours at Exit 1 (Figure~\ref{fig:departure_rate_exit1}), and from 7 to 8 hours at Exit 2 (Figure~\ref{fig:departure_rate_exit2}). Additionally, the departure rate curve becomes smoother. As depicted in Figure~\ref{fig:departure_rate_exit1}, with heightened anxiety aversion, the maximum gap between the number of vehicles departing in three consecutive time periods decreases from 15 to 11, and then to 9 at Exit 1. Similarly, at Exit 2, this maximum gap reduces from 16 to 11, and subsequently to 6.

\subsection{Effectiveness of BNE Algorithms}

In this section, we verify the effectiveness of proposed NE solving algorithms. According to the definition of BNE \eqref{eq:BNE_definition}, no one can unilaterally change the strategy to raise the expectation of payoff. In this experiment, we allow each EV to change his strategy in the BNE obtained by Algorithm~\ref{alg:solve_BNE_normal} and Algorithm~\ref{alg:solve_BNE_traffic} and calculate the cost expectation before and after the change. For the parameter settings, we use the A1 highway traffic crossing the CS in Experiment~\ref{ssec:EV_cost}. We take the BNE at 3:00 PM when the traffic is heavy. The results are shown in Figure~\ref{fig:NE_effective}.

\begin{figure*}[t]
\centering
\begin{subfigure}[t]{0.5\textwidth}
        \centering
  \includegraphics[width=0.9\linewidth]{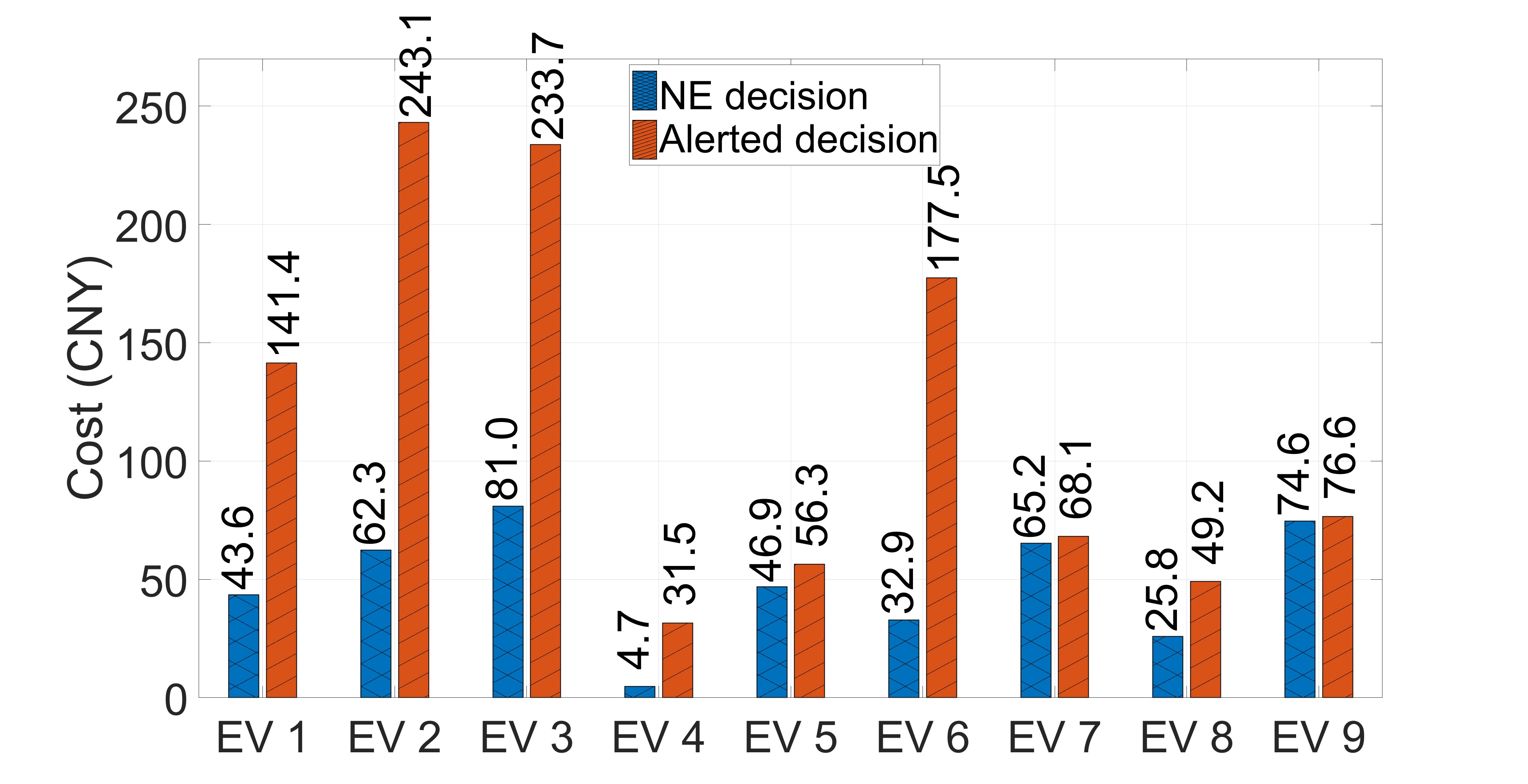}
  \caption{Without destination distribution knowledge (Algorithm~\ref{alg:solve_BNE_normal}). \\NE decisions: EV 1,2,3,6,8,9 charge, EV 4,5,7 not charge.}
  \label{fig:NE_effective_normal}
\end{subfigure}%
\hfill
\begin{subfigure}[t]{0.5\textwidth}
        \centering
  \includegraphics[width=0.9\linewidth]{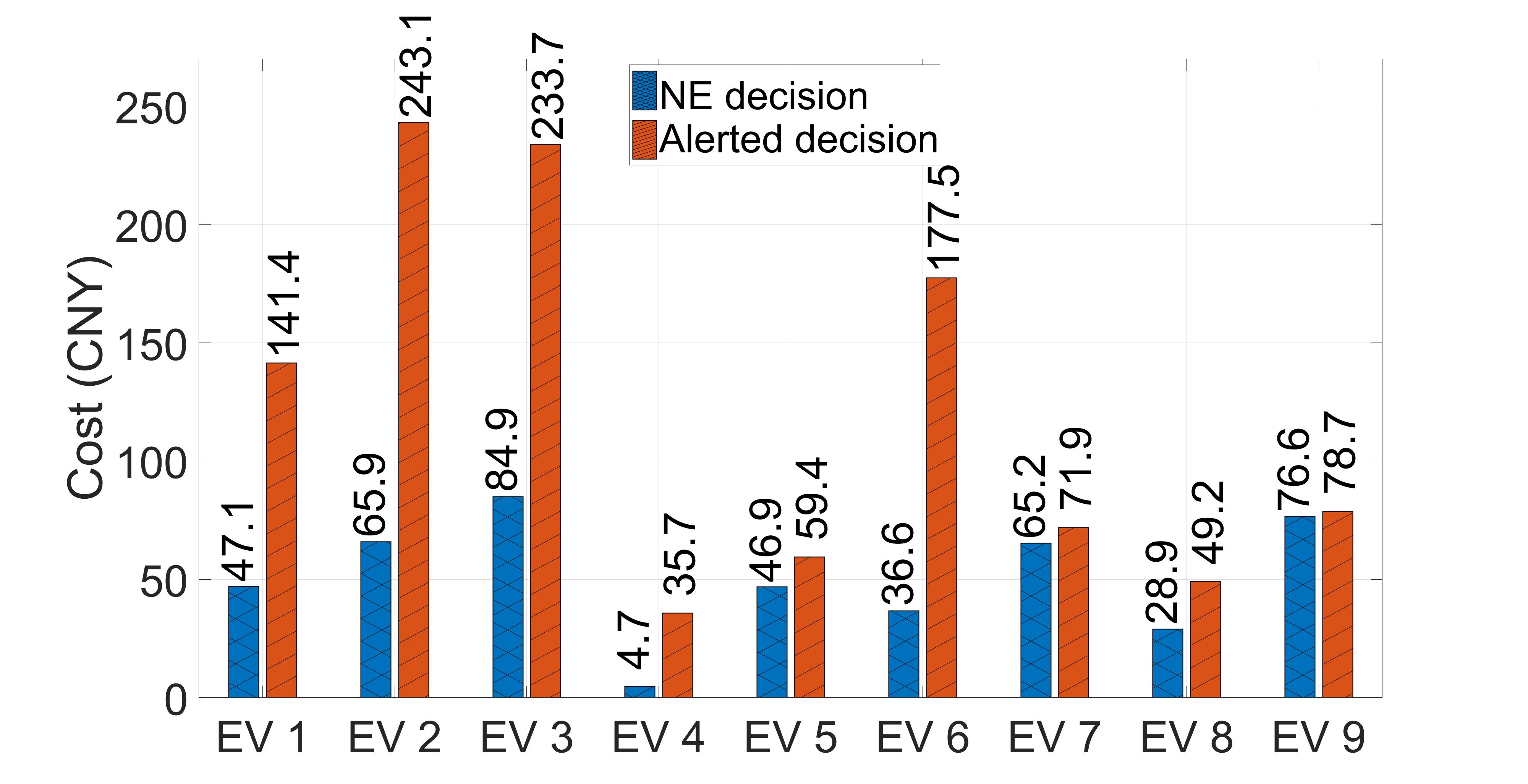}
  \caption{With destination distribution knowledge (Algorithm~\ref{alg:solve_BNE_traffic}). \\NE decisions: EV 1,2,3,6,8 charge, EV 4,5,7,9 not charge.}
  \label{fig:NE_effective_traffic}
\end{subfigure}%
\caption{\textcolor{black}{Effectiveness of NE Algorithm.}}
\label{fig:NE_effective}
\end{figure*}

As shown in Figure~\ref{fig:NE_effective}, each EV has an increase in cost after the change decision. This shows that the BNE obtained by our algorithm is effective.

\subsection{Benefits of accurate prediction}
\label{ssec:benefits_cs}
In Section~\ref{sec:benefit_cs}, we show the impact of accurate predictions of EV behavior on charging station revenues. In this section, we compare charging stations using a different prediction method to ours.

\begin{figure*}[t]
\centering
\begin{subfigure}[t]{0.33\textwidth}
  \includegraphics[width=0.9\linewidth]{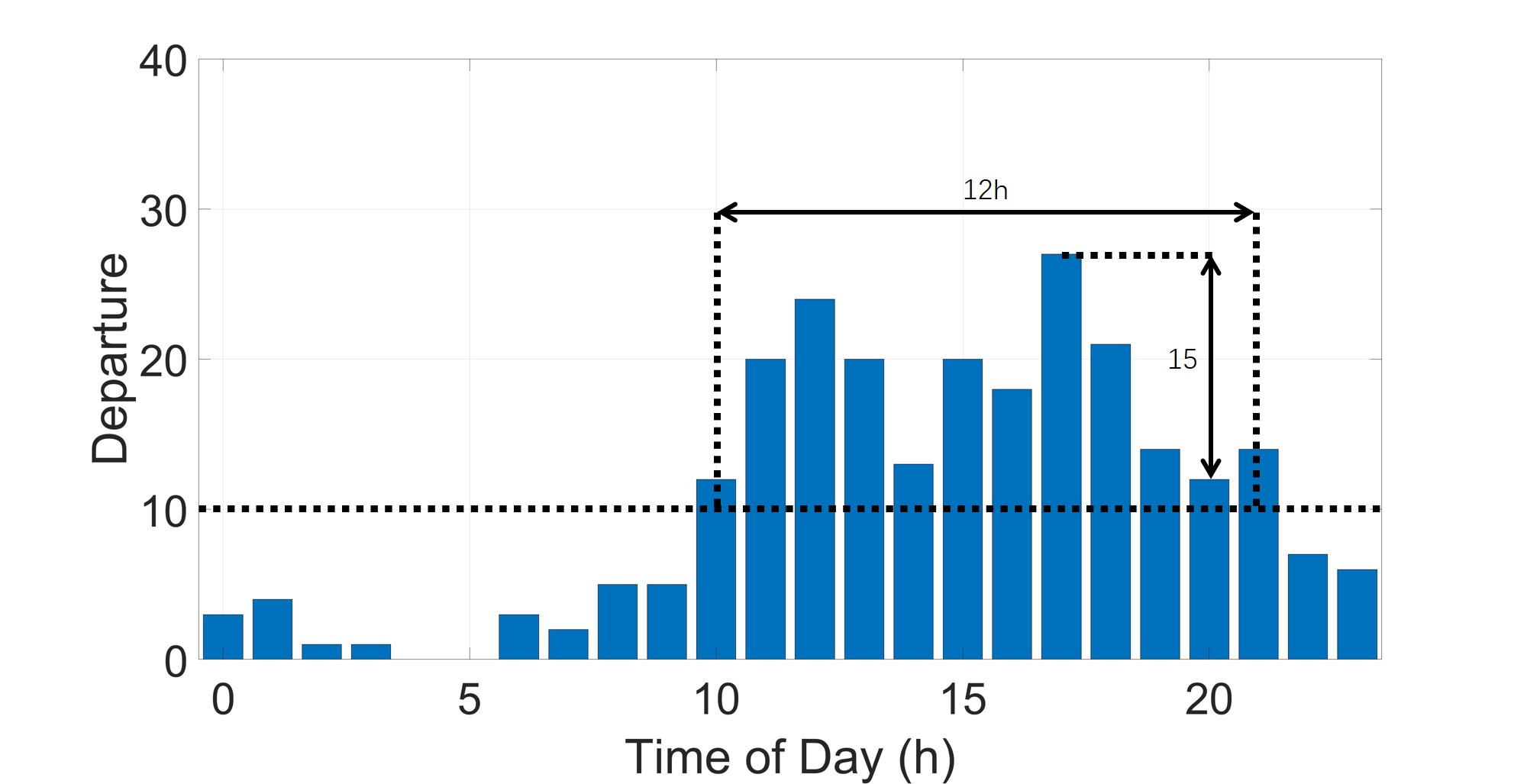}
  \caption{$\lambda=0$}
  \label{fig:departure_rate_exit2_no}
\end{subfigure}%
\hfill
\begin{subfigure}[t]{0.33\textwidth}
  \includegraphics[width=0.9\linewidth]{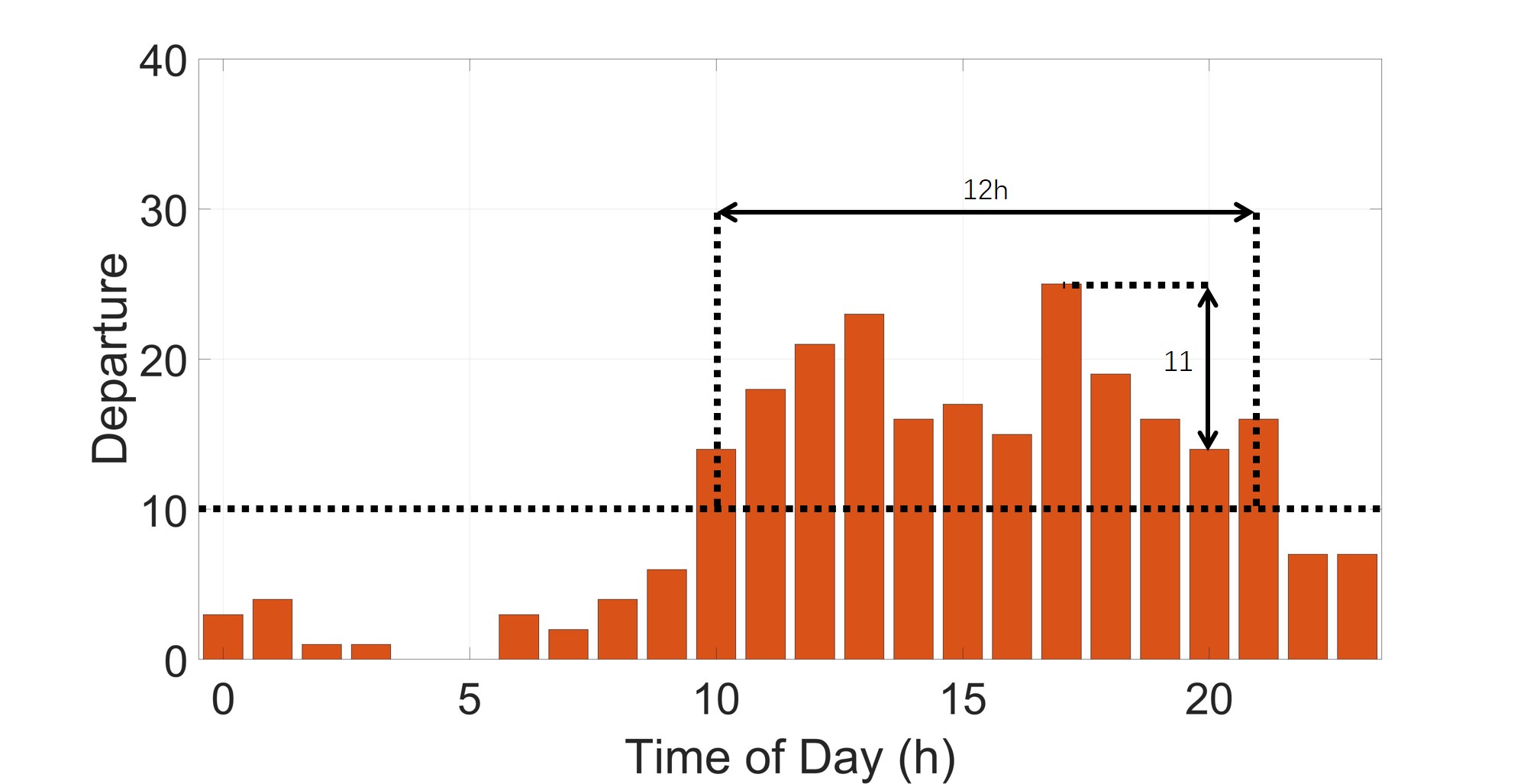}
  \caption{$\lambda=1.5$}
  \label{fig:departure_rate_exit2_low}
\end{subfigure}%
\hfill
\begin{subfigure}[t]{0.33\textwidth}
  \includegraphics[width=0.9\linewidth]{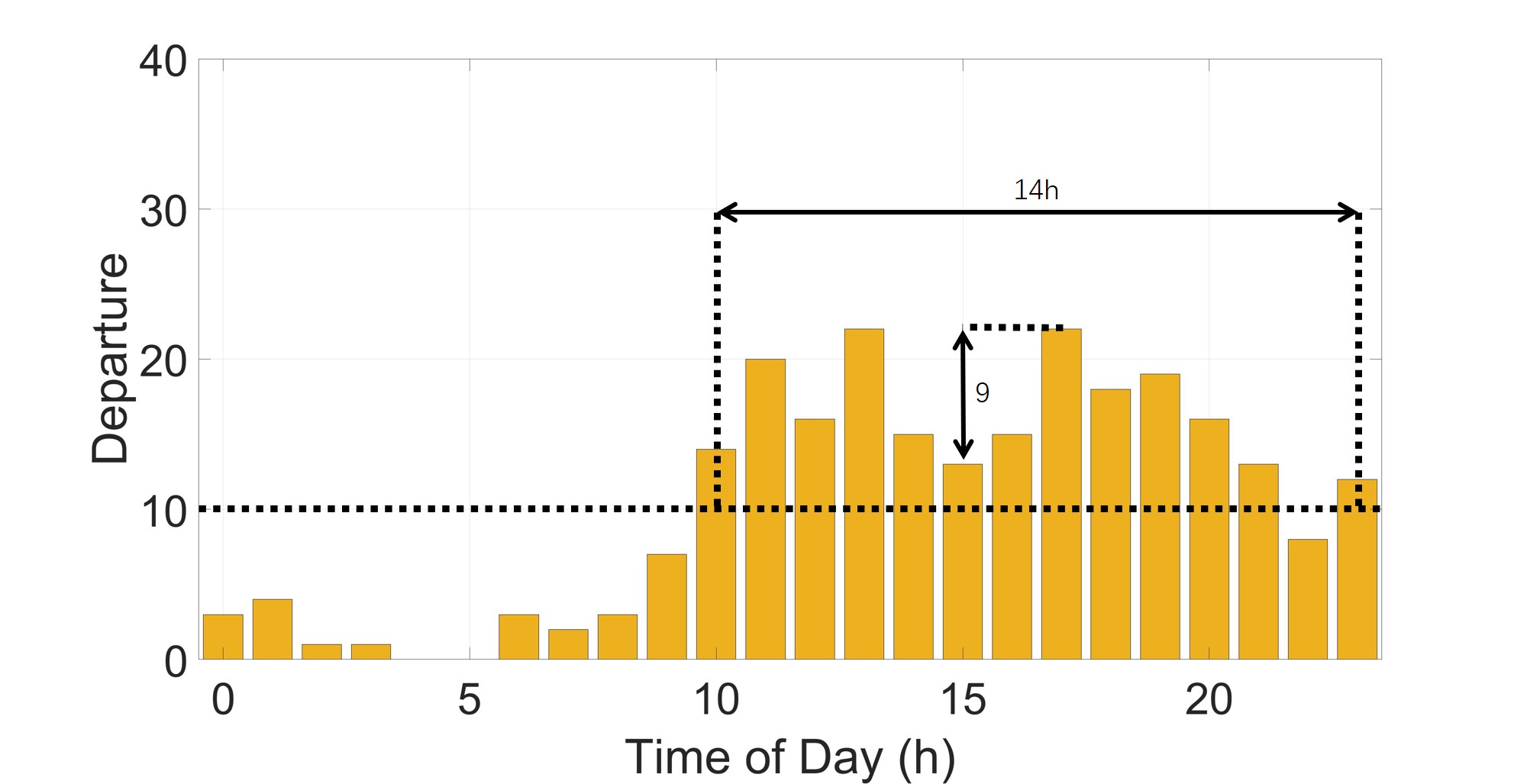}
  \caption{$\lambda=3.5$}
  \label{fig:departure_rate_exit2_high}
\end{subfigure}%
\caption{Departure rate of exit 2 as rationality varies.}
\label{fig:departure_rate_exit2}
\end{figure*}

\begin{figure*}[t]
\centering
\begin{subfigure}[t]{0.33\textwidth}
  \includegraphics[width=0.9\linewidth]{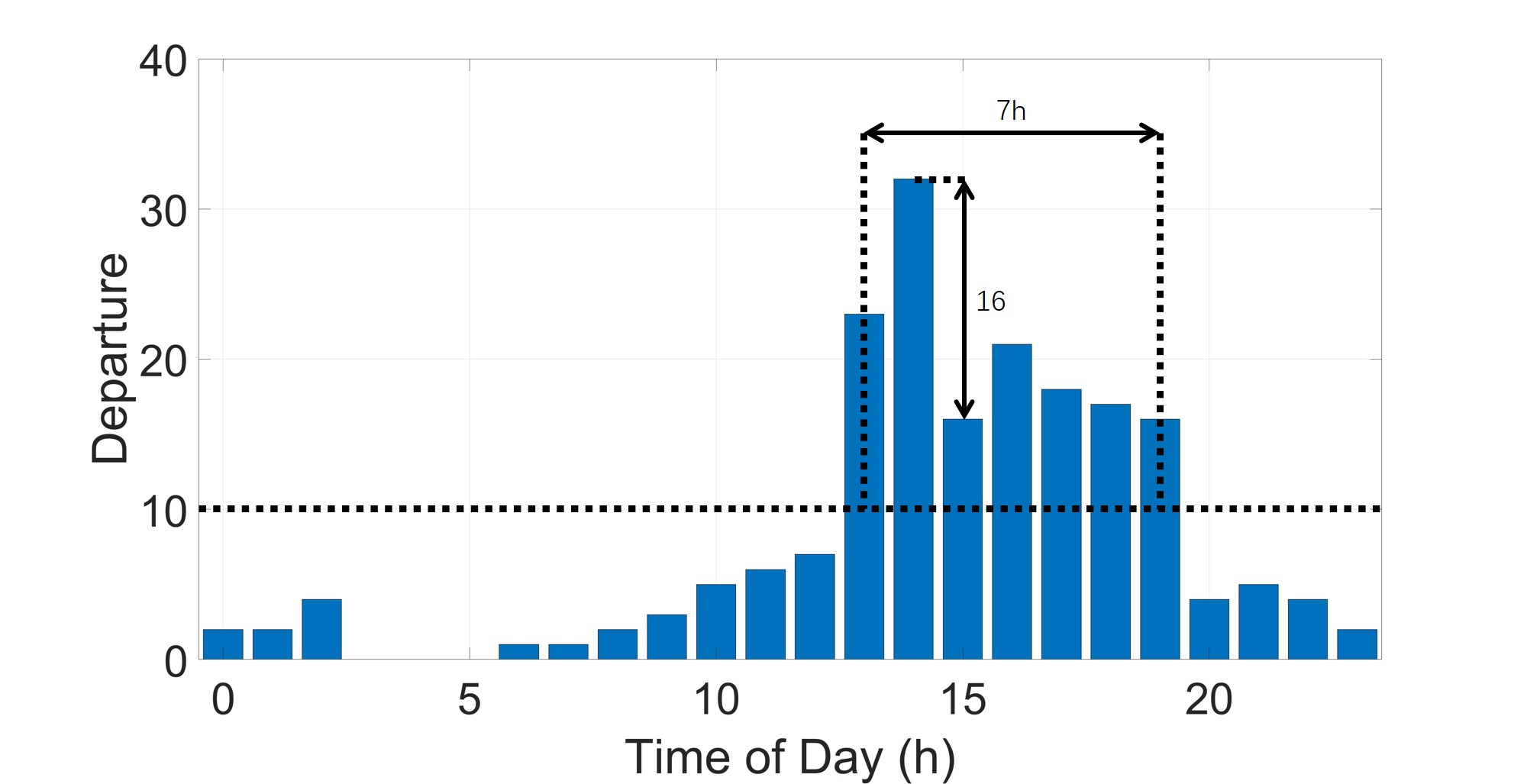}
  \caption{$\lambda=0$}
  \label{fig:departure_rate_exit1_no}
\end{subfigure}%
\hfill
\begin{subfigure}[t]{0.33\textwidth}
  \includegraphics[width=0.9\linewidth]{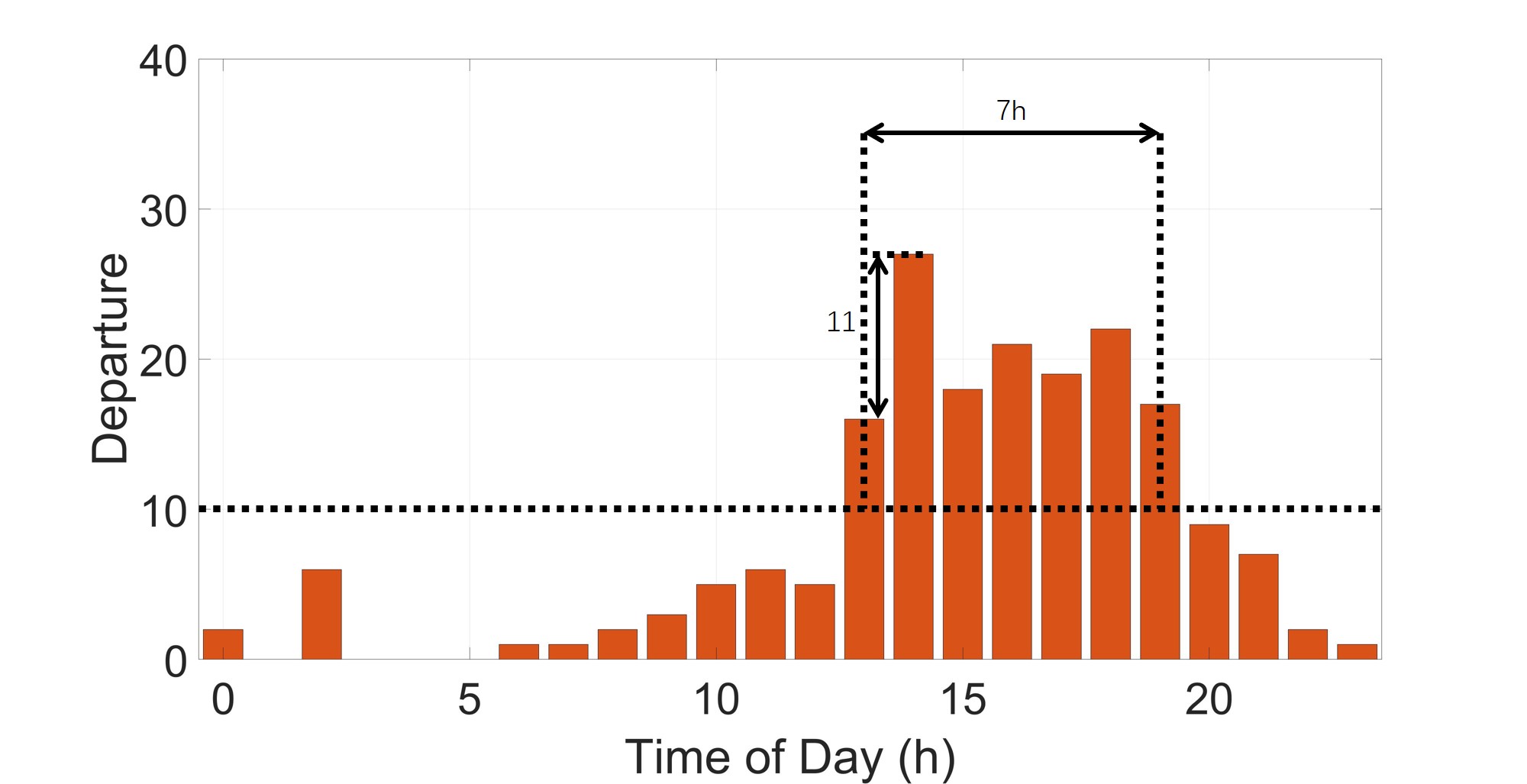}
  \caption{$\lambda=1.5$}
  \label{fig:departure_rate_exit1_low}
\end{subfigure}%
\hfill
\begin{subfigure}[t]{0.33\textwidth}
  \includegraphics[width=0.9\linewidth]{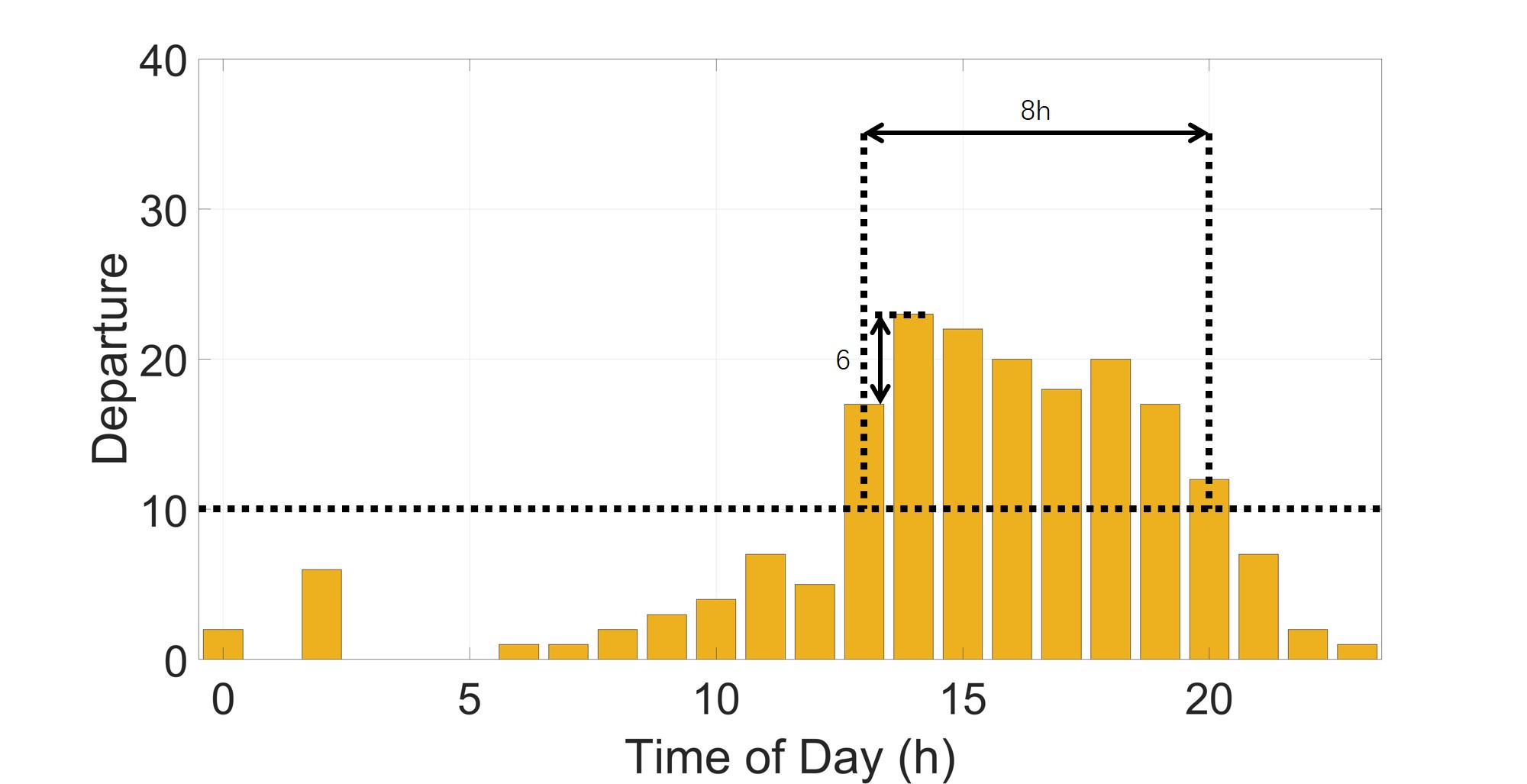}
  \caption{$\lambda=3.5$}
  \label{fig:departure_rate_exit1_high}
\end{subfigure}%
\caption{Departure rate of exit 1 as rationality varies.}
\label{fig:departure_rate_exit1}
\end{figure*}

In this experiment, we employed our proposed method alongside various benchmarks as predictive values to determine the optimal day-ahead charging strategy for the charging station, as described by Equation~\eqref{eq:CS_optimization}. \color{black} In this simulation, we use \textit{linprog} toolbox in MATLAB to solve the linear optimization problem \eqref{eq:CS_optimization}. \color{black} The range anxiety parameter was set in accordance with the benchmark established by \cite{hu2019modeling}. Subsequently, we calculated the average daily charging cost of the charging station under different day-ahead power purchase strategies. In scenarios where the charging demand could not be met, we assumed that the charging station would procure power in the real-time market as an emergency measure. The real-time market electricity price was assumed to be 1.2 times the day-ahead price, with an additional emergency processing fee of 50 CNY per visit. For the electricity tariff setting $Pe$, we utilized the Singapore vesting contract energy price from July 8, 2024 \cite{SingaporeElectricityPrice}. \color{black} Under above parameters, we calculate the optimal charging decisions with Equation~\eqref{eq:CS_optimization} corresponding to different predicted demand $D_t$ under the proposed methods and benchmarks. Then, we calculate the daily charging cost $\mathbf{X}^\intercal \mathbf{Pe}$ under different prediction methods. \color{black} The results of the daily charging cost analysis are presented in Figure~\ref{fig:cost_cs}.

\begin{figure*}[t]
\centering
\begin{subfigure}[t]{0.5\textwidth}
    \centering
    \includegraphics[width=0.9\linewidth]{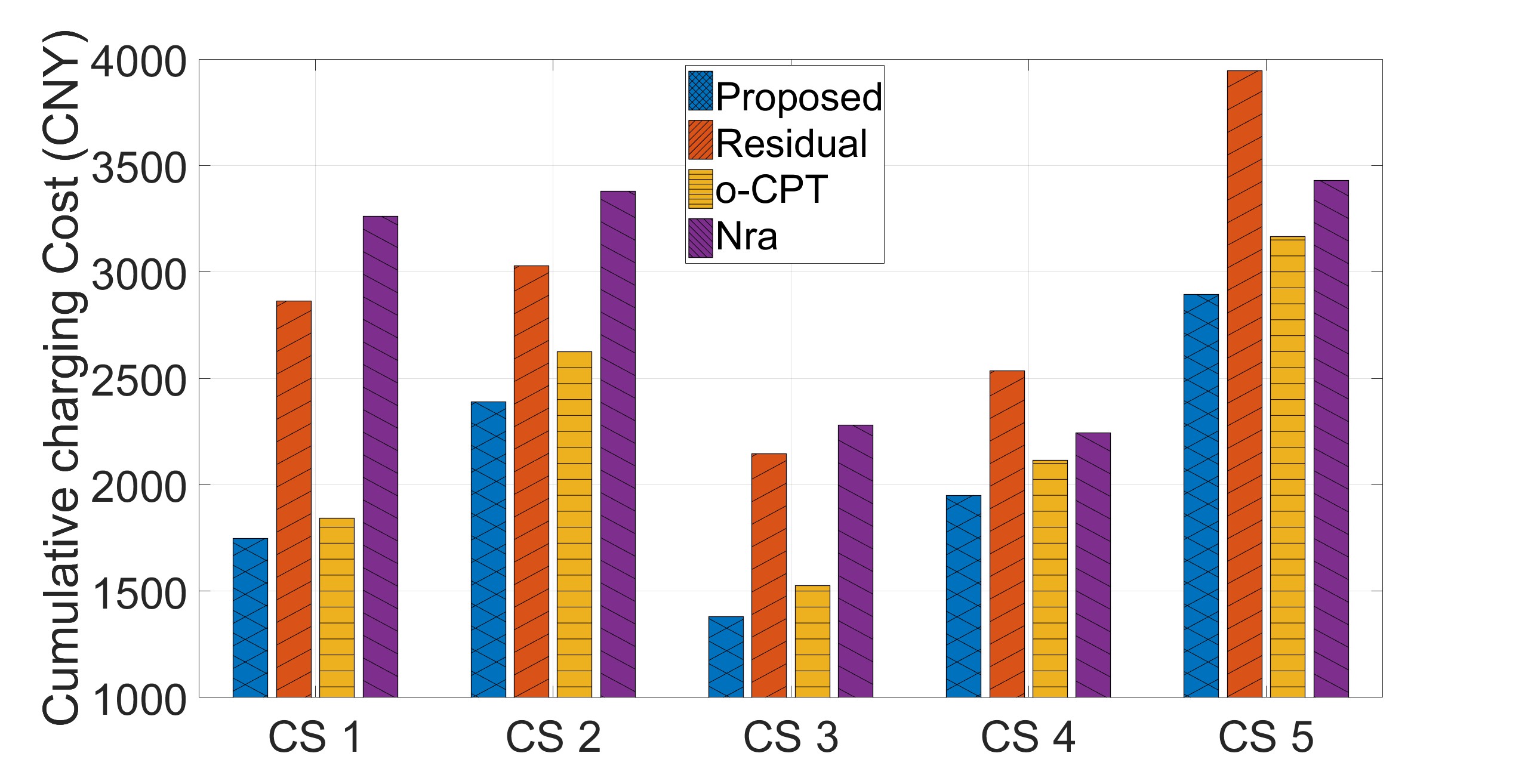}
    \caption{Weekdays.}
    \label{fig:cost_cs_weekdays}
\end{subfigure}%
\hfill
\begin{subfigure}[t]{0.5\textwidth}
    \centering
  \includegraphics[width=0.9\linewidth]{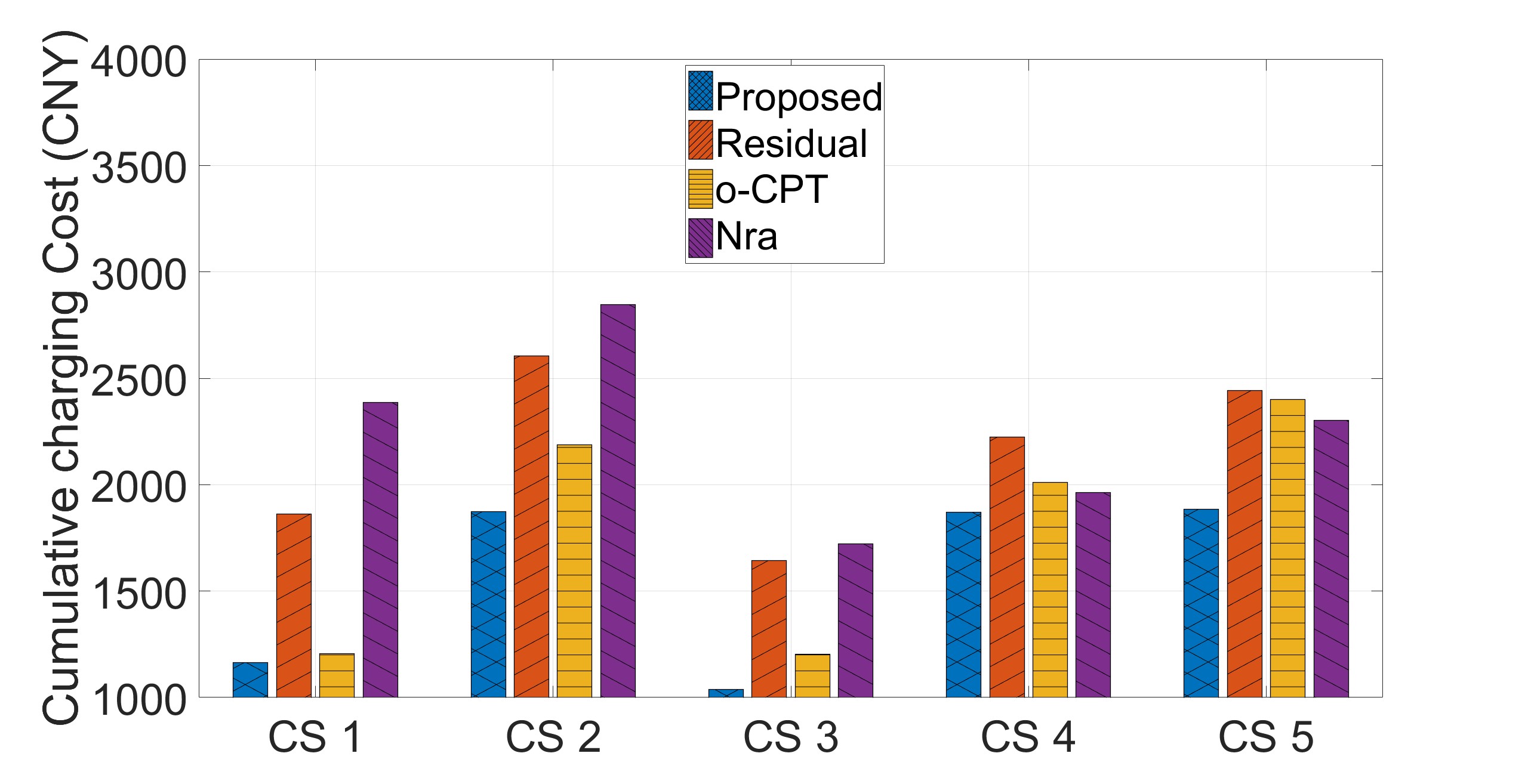}
  \caption{Weekends.}
  \label{fig:cost_cs_weekends}
\end{subfigure}%
\caption{Daily charging cost with different prediction methods.}
\label{fig:cost_cs}
\end{figure*}

As illustrated in Figure~\ref{fig:cost_cs}, accurate forecasting methodologies consistently reduce charging costs across all charging stations, thereby enhancing revenues. For instance, at CS 2, the charging cost predicted using our method is approximately 1,873 CNY, which represents 71\% of the cost incurred using the Residual method (2,605 CNY), 85\% of the cost using the o-CPT method (2,186 CNY), and 65\% of the cost using the Nra method (2,864 CNY).

Furthermore, we observe that the o-CPT method can approximate the effectiveness of our method in scenarios involving charging stations with a limited number of vehicles, such as CS 1. However, our method significantly outperforms the o-CPT method at stations with higher vehicle charging demands, such as CS 5. This observation aligns with our previous analysis, as the o-CPT method optimizes for a single vehicle and does not account for interactions when multiple vehicles make concurrent decisions.

\section{Conclusions and Future Directions}
\label{sec:conclusion}

This paper propose a new prospect theoretical framework to analyze the EV highway charging behaviors. We propose a Bayesian game model where EV drivers seek to minimize costs, including the charging fee, range anxiety, and queuing time. We prove the existence and uniqueness of the Bayesian Nash equilibrium, and propose a search algorithm to solve the equilibrium. Our simulations demonstrate that the range anxiety and the risk aversion tendency have significant impact on the charging decision and charging station status.

As for further directions, one may consider the optimal site selection issues of the highway charging station based on this result. Another interesting topic may be the optimal pricing on the highway charging stations.

\section*{Code Availability}
The complete MATLAB source code implementing the proposed bounded rationality highway charging Bayesian game framework, including all Bayesian Nash Equilibrium solvers (with and without destination distribution knowledge), benchmark methods, and figure-generation scripts, is openly available on GitHub:
\begin{center}
\url{https://github.com/T-Lab-CUHKSZ/-IOTJ25-Incorporating-Bounded-Rationality-into-Electric-Vehicle-Highway-Charging-Decisions}
\end{center}

\appendix
\subsection{Proof of Lemma~\ref{lmm:NE_form} and Theorem~\ref{thm:existence_uniqueness}}
\label{apdx-prof_NE_form}
\begin{proof}
We first prove Lemma~\ref{lmm:NE_form} by analyzing the monotonicity of function $h(SoC)$. Due to $A_i$'s \textcolor{black}{definition} in Equation~\eqref{eq:a_i}, we discuss in three categories based on the relationship of $SoC_i$, $SoC_{d,i}$ and $SoC_{s,i}$.

When $SoC_i < SoC_{d,i}$, $\mathbb{E}[\Pi(x_i=0|SoC_i)] = - \infty$, $h(SoC_i) < 0$.

When $SoC_i > SoC_{s,i}$, $E[\Pi(x_i=0|SoC_i] = 0$,
\begin{align}
    h(SoC_i) = 0 + V_i (SoC_t-{SoC}_i) \cdot P + \mu_{q,i} \cdot \mathbb{E}[t] > 0. \nonumber
\end{align}

Then we discuss the case when $SoC_{d,i} \leq SoC_i \leq SoC_{s,i}$. Taking the derivation, we have
\begin{align}
\small
&\frac{\rm d h(SoC_i)}{\rm d SoC_i} = -\alpha_i \mu_a \lambda_i (SoC_{s,i}-SoC_i)^{\alpha_i-1}-V_iP. \\
&\frac{\rm d^2 h(SoC_i)}{(\rm d SoC_i)^2} = -\alpha_i (\alpha_i-1)\mu_a \lambda_i (SoC_{s,i}-SoC_i)^{\alpha_i-2}.
\end{align}
Since $\alpha_i > 1$,  $\frac{\rm d^2 h(SoC_i)}{(\rm d SoC_i)^2} < 0$. Thus for any $SoC_i$ in $[SoC_{d,i},SoC_{s,i}]$, $\frac{\rm d h(SoC_i)}{\rm d SoC_i} > \frac{\rm d h(SoC_{s,i})}{\rm d SoC_i} = -V_iP $. Function $h(SoC_i)$ is either decreasing in $[SoC_{d,i},SoC_{s,i}]$ (case \normalsize{\textcircled{\scriptsize{1}}}\normalsize, or first increasing then decreasing (case \normalsize{\textcircled{\scriptsize{2}}}\normalsize\enspace or \normalsize{\textcircled{\scriptsize{3}}}\normalsize, depending on the sign of $h(SoC_{d,i})$), as shown in Figure~\ref{fig:h_shape}.

\begin{figure}[ht]
    \centering
    \includegraphics[width=0.8\linewidth]{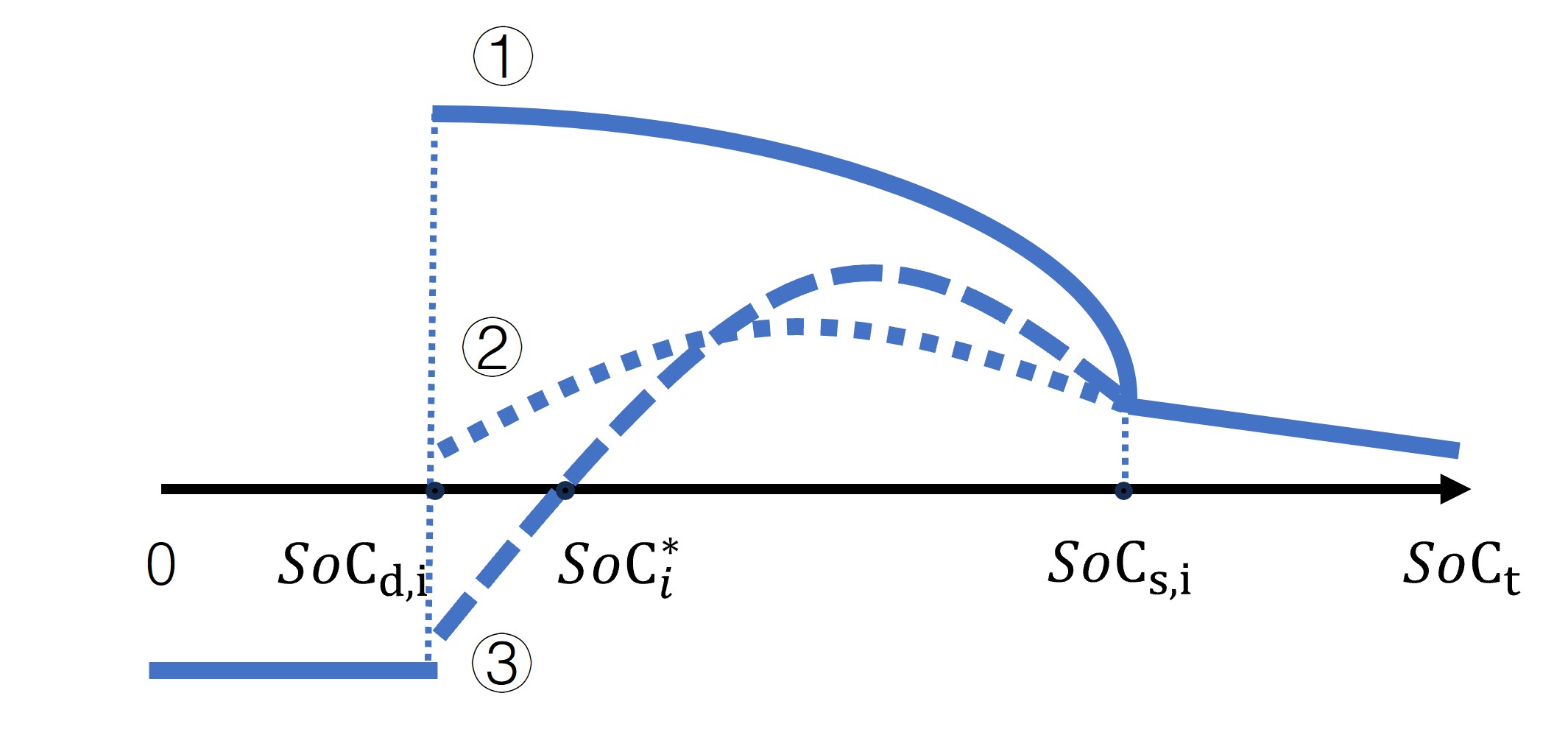}
    \caption{There possible shapes for function $h(SoC)$.}
    \label{fig:h_shape}
\end{figure}

The BNE form in Equation~\eqref{eq:BNE_form} is equivalent to
\begin{align} \label{eq:BNE_form_h}
\small
    h(SoC_i) < 0 \text{ if } SoC_i < \widehat{SoC_i}, \nonumber \\
    h(SoC_i) \geq 0 \text{ if } SoC_i \geq \widehat{SoC_i}.
\end{align}
The function $h(SoC)$ is negative when $SoC_i < SoC_{d,i}$, and is positive when $SoC_i > SoC_{s,i}$. If $h(SoC_{d,i}) \geq 0$ (case \normalsize{\textcircled{\scriptsize{1}}}\normalsize\enspace or \normalsize{\textcircled{\scriptsize{2}}}\normalsize), point $\widehat{SoC_i} = SoC_{d,i}$ is the only point that satisfies condition~\eqref{eq:BNE_form_h}; If $h(SoC_{d,i}) < 0$ (case \normalsize{\textcircled{\scriptsize{3}}}\normalsize), since $h(SoC_i)$ is continuous in $[SoC_{d,i},SoC_{s,i}]$, according to intermediate value theorem, there exist a point $SoC^*_i$ in the increasing interval satisfies $h(SoC^*_i)=0$. Point $\widehat{SoC_i} = SoC^*_i$ is the only point satisfies condition~\eqref{eq:BNE_form_h}. Therefore, in every case, there is a unique value that satisfies Equation~\eqref{eq:BNE_form_h}, which complete the proof of Lemma~\ref{lmm:NE_form}.

Lemma~\ref{lmm:NE_form} also shows that in a system of N players, each player has a unique decision point. In a $N$-player system, each player makes charging decisions based on his own unique decision point. This also implies that there is a unique Bayesian equilibrium in the system, completing the proof of Theorem~\ref{thm:existence_uniqueness}.
\end{proof}

\bibliographystyle{IEEEtran}
\bibliography{ref.bib}

\begin{IEEEbiography}[{\includegraphics[width=1in,height=1.25in,clip,keepaspectratio]{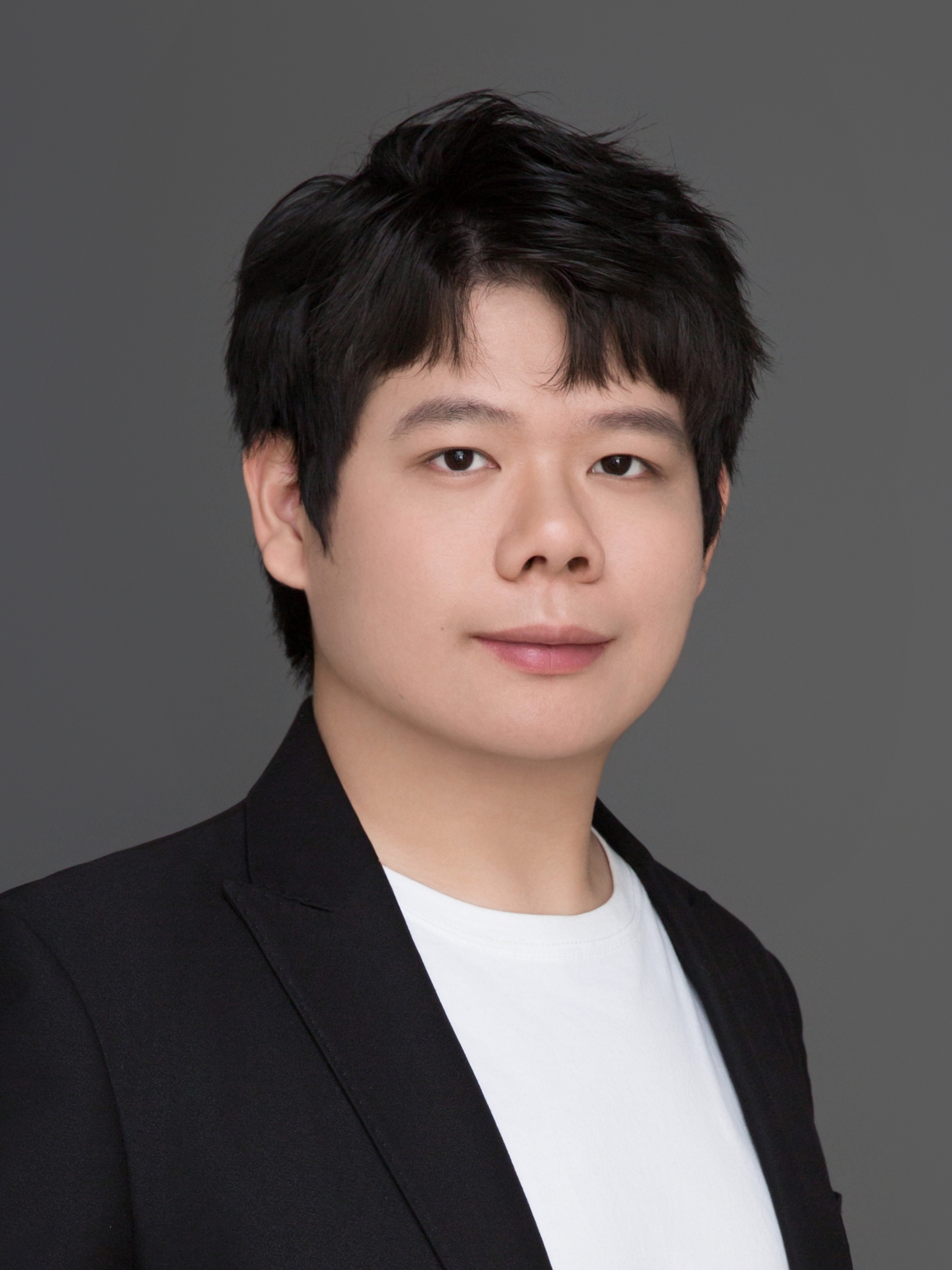}}]{Huanyu Yan}
Huanyu Yan obtained the B.Eng. degree in Computer Science and Technology from the Northwestern Polytechnical University, Shaanxi, China, in 2020. He is currently pursuing the Ph.D. degree in Computer and Information Engineering with the Chinese University of Hong Kong (Shenzhen), Shenzhen, China. His research interests include electric vehicle charging scheduling, charging station pricing and optimizations, applications of game theory and economic mechanisms.
\end{IEEEbiography}

\begin{IEEEbiography}[{\includegraphics[width=1in,height=1.25in,clip,keepaspectratio]{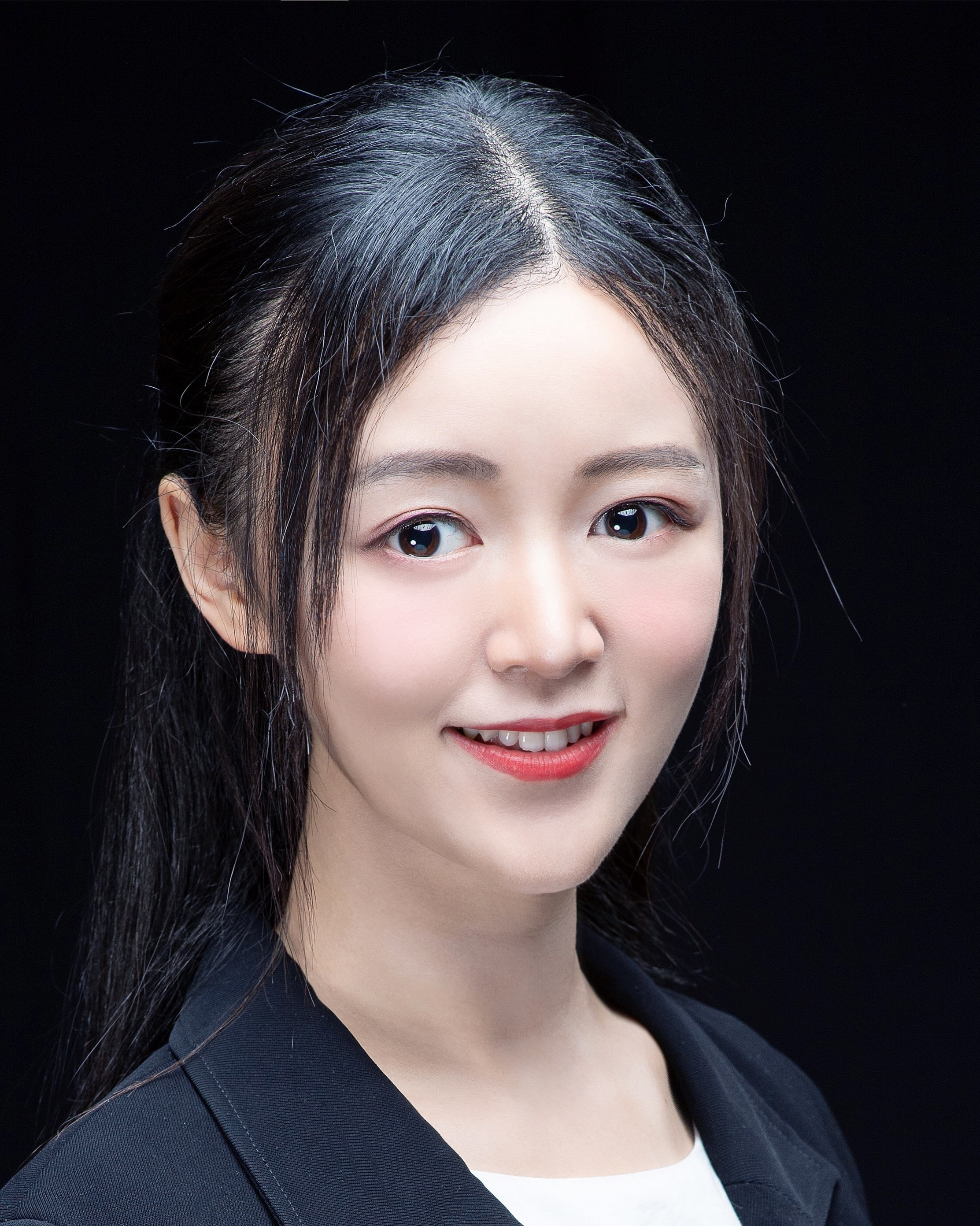}}]{Xiaoying Tang}
Xiaoying Tang (S'13-M'16) is currently a tenure-track assistant professor in The Chinese University of Hong Kong (Shenzhen). Before that, she worked as a Research Scientist at EPFL and CUHK during 2016-2019, respectively. She received her Ph.D. degree from CUHK in 2016. Her current research interests include federated learning, LLM, optimization algorithm design for EV charging and other information networks. Dr. Tang received the Best Paper Award of IEEE SmartGridComm 2013 and IEEE SmartGridComm 2024, Excellent Poster Presentation Award of IEEE ICCT 2024 and Distinguished Paper Award of IEEE/ACM ASE 2023. She has a used name: Wanrong Tang.
\end{IEEEbiography}

\end{document}